\documentclass[10pt]{article}

\usepackage[utf8]{inputenc} 
\usepackage[T1]{fontenc}    
\usepackage{hyperref}       
\usepackage{url}            
\usepackage{booktabs}       
\usepackage{amsfonts}       
\usepackage{nicefrac}       
\usepackage{microtype}      
\usepackage{xcolor}         

\usepackage{authblk}
\usepackage{amsmath}
\usepackage{amsthm}
\usepackage{graphicx} 
\usepackage{bm}
\usepackage{lipsum}
\usepackage{tabulary}
\usepackage{booktabs}
\usepackage{xcolor}
\usepackage{pifont}
\usepackage{setspace}
\usepackage{subfigure}
\usepackage{cleveref}
\usepackage{thmtools,thm-restate}
\usepackage[subtle,tracking=normal, wordspacing=normal,mathdisplays=tight]{savetrees}

\usepackage{geometry}
\newtheorem*{theorem*}{Theorem}
\newtheorem{theorem}{Theorem}[section]

\newtheorem{claim}[theorem]{Claim}

\newtheorem{definition}[theorem]{Definition}
\newtheorem{observation}[theorem]{Observation}

\newcommand{\B}{\beta}
\newcommand{\C}{\hat{\beta}}
\newcommand{\BC}{\beta\hat{\beta}}
\newcommand{\deltaw}{\delta}

\newcommand{\yuxuan}[1]{}
\newcommand{\hzh}[1]{}
\newcommand{\yk}[1]{}
\newcommand{\gyk}[1]{}

\newcommand{\E}{\mathrm{E}}

\newcommand{\argmax}{\mathop{\arg\max}}

\newcommand{\surp}{Surp}

\newcommand{\win}{\mathrm{G}}
\newcommand{\lose}{\mathrm{S}}

\newcommand{\bel}{B}
\newcommand{\belset}{\mathcal{\bel}}

\title{How Gold to Make the Golden Snitch: Designing the ``Game Changer'' in Esports}
\author[1]{Zhihuan Huang\thanks{Co-first author, equal contribution}$^,$}
\newcommand\CoAuthorMark{\footnotemark[\arabic{footnote}]}
\author[1]{Yuxuan Lu\protect\CoAuthorMark$^,$}
\author[1]{Yongkang Guo}
\author[1]{Yuqing Kong}
\affil[1]{Peking University}
\affil[1]{\texttt{\{zhihuan.huang, yx\_lu, yongkang.guo, yuqing.kong\}@pku.edu.cn}}

\begin{document}

\maketitle

\begin{abstract}

Many battling games utilize a special item (e.g. Roshan in Defense of the Ancients 2 (DOTA 2), Baron Nashor in League of Legends (LOL), Golden Snitch in Quidditch) as a potential ``Game Changer''. The reward of this item can enable the underdog to make a comeback. However, if the reward is excessively high, the whole game may devolve into a chase for the ``Game Changer''. Our research initiates with a Quidditch case study, a fictional sport in Harry Potter series, wherein we architect the Golden Snitch's reward to maximize the audience's surprise. Surprisingly, we discover that for equally competent teams, the optimal Snitch reward is zero. Moreover, we establish that under most circumstances the optimal score aligns with the game's expected duration multiplied by the teams' strength difference. Finally, we explore the correlation between the ``Game Changer's'' reward and audience surprise in Multiplayer Online Battle Arena (MOBA) games including DOTA 2 and LOL, finding that the optimal reward escalates with increasing team strength inequality.
\end{abstract}

\section{Introduction}
\label{sec:intro}
Dating back to antiquity, the excitement of the arena has captivated audiences far and wide. The public nature of victory and defeat in the arena fosters a shared experience between players and audiences alike. Today, thanks to live-streaming technology, digital arenas for traditional sports and eSports have emerged, allowing billions to share the arena experience in real-time, driving a huge amount of revenue~\cite{website:invenglobal}.

Audience engagement in competitive events thrives on fluctuating beliefs about the potential victor. The level of belief fluctuation can be described by an indicator of surprise~\cite{ely2015suspense}. Empirical evidence verifies that people derive entertainment from surprise~\cite{ijcai2021-36,lanceleague,bizzozero2016importance,buraimo2020unscripted,scarf2019outcome}. Hence, lopsided contests, often lack of surprise, may fail to ignite excitement. This is why competition systems tend to match players with comparable skill levels. However, player constraints sometimes make this ideal impossible to realize, result in unbalanced match-ups.

In response to lopsided match-ups, many battling games set a special item as a potential ``Game Changer''. A fictional sport, Quidditch, invented by J.K. Rowling in her iconic \emph{Harry Potter} series, has a ``Game Changer'', the Golden Snitch. In Quidditch, similar to soccer, teams score points (10 each) by throwing a leather ball into the opponent's hoop. Besides, each team has a special player known as the seeker, whose mission is to catch the Golden Snitch, a feat valued 150 points. The game ends once the Golden Snitch is caught, and the team with the higher points wins. Unsurprisingly, the high-stakes chase for the Golden Snitch has always been the most exciting scene in Quidditch.
Two popular Multiplayer Online Battle Arena (MOBA) games, Defense of the Ancients 2 (DOTA 2) and League of Legends (LOL), also incorporate a ``Game Changer''. These games pit two teams of five players against each other, with participants accumulating wealth to equip their characters for battle. Victory depends on destroying the enemy team's main building. In DOTA 2 and LOL, the ``Game Changers'' are respectively Roshan and Baron Nashor. The team that wins ``Game Changer'' will gain substantial advantages. Like Quidditch, these pivotal moments often attracts the audience's attention~\cite{ijcai2021-36}.

It is natural to ask how should the reward for the ``Game Changer'' be structured. Intuitively, it should surpass a normal item in value, but setting it too high risks devolving the entire match-up to a race for the ``Game Changer''. Therefore, a guideline for reward design is crucial. Although Quidditch is a fictional sport, it captures various real-world sports characteristics. Consequently, it serves as an starting point for our theoretical exploration into reward design.

\paragraph{This paper} This paper theoretically and empirically investigates the ``Game Changer'' design. In \Cref{sec:prob}, we model the Quidditch game as a time-homogeneous Markov chain. \Cref{sec:quidditch} employs the surprise indicator~\cite{ely2015suspense} as the optimization goal, deriving the optimal score of Golden Snitch that maximizes the expected surprise. \Cref{sec:moba} generalizes the model to MOBA games, conducts empirical analysis on DOTA 2 and LOL, and computes the optimal rewards for Roshan and Baron Nashor.

\paragraph{Results overview} Here we describe the paper's primary insights. Within the Quidditch context, in each round, the first team has a scoring probability of $p$, while the second team has a $1-p$ chance. The probability that any round concludes the game (due to the Golden Snitch) is denoted as $q$. We formulate a closed-form expression for the expected surprise, yielding the subsequent results:
\begin{itemize}
    \item \textbf{Near-balanced match-up}:
    When $p=0.5$, indicating a balanced match-up, we theoretically show that the optimal score for the Golden Snitch is zero, regardless of the difficulty of the Snitch ($\forall q\in (0,1)$). In a nearly balanced match-up ($p\approx 0.5$), as long as $q$ is not negligible, numerical studies show that the optimal Golden Snitch score remains zero.
    \item \textbf{Unbalanced match-up}: 
    For an unbalanced match-up ($p<0.5-\epsilon$), when $q$ is sufficiently small, the optimal score for the Golden Snitch approximates to $\frac{1}{2q}\left(\frac{1-p}{p}-1\right)$. This corresponds to the product of the expected number of rounds and the difference in team strengths.
    \item \textbf{General case} In general case, though the optimal $x^*$ does not have an analytic solution, we provide a closed-form upper-bound and numerically show that the optimal $x^*$ is either $0$ or close to the upper-bound.
\end{itemize}

Regarding MOBA, we utilize empirical data from LOL and DOTA 2 to calibrate the model parameters. Through numerical studies, we explore the relationship between the optimal reward of the ``Game Changer'' and the strength gap between the two teams. For balanced match-ups, the optimal reward is roughly the wealth income in 2 to 3 rounds. In unbalanced match-ups, larger strength gap leads to greater optimal reward. Therefore, in both Quidditch and MOBA, our findings suggest that a more unbalanced match-up requires a higher optimal reward for ``Game Changer''. For a balanced match-up, the optimal reward is relatively small for MOBA.

\subsection{Related Work}

Starting from Ely et al.~\cite{ely2015suspense} which provide a clear model definition of suspense and surprise, a series of works study the effect of surprise on people's watching experience, including not only traditional sports such as tennis~\cite{bizzozero2016importance}, soccer~\cite{buraimo2020unscripted,lucas2017goaalll} and rugby~\cite{scarf2019outcome}, but also eSports like LOL~\cite{ijcai2021-36,lanceleague}. These empirical results show that the audience's perceived entertainment utility of a game is positively correlated with the amount of surprise. However, none of those works focus on the theoretical design of the reward of ``Game Changer'' in the battling game. Ely et al.~\cite{ely2015suspense} also study how to design the number of rounds to maximize surprise. However, in many battling games, it's easier to add a ``Game Changer'' than change the number of rounds. 

Huang et al.~\cite{huang2021bonus} consider a game with a fixed number of rounds and study how to design the final round's score to maximize surprise. Despite the analogous optimization objective, we consider a more complicated game model with indeterminate rounds and use different techniques. Our techniques can be generalized to games that can be modeled as time-homogeneous Markov chains. Huang et al.~\cite{huang2021bonus} show that their optimal score is approximately ``expected leads'', the points the weaker team needs to come back in the final round. Similarly, our results also show that when the game lasts longer and the strength difference is larger, the optimal score should be higher. 

Several other works~\cite{brams2018making,braverman2008mafia} focus on how to tune the rules to improve the fairness of the game. This work focuses on the overall surprise from the side of the audience and there may be a trade-off between surprise and fairness as Huang et al.~\cite{huang2021bonus} show that a more surprising game may be less balanced for the two players. 

Regarding modeling MOBA games, Yang et al.~\cite{yang2014identifying} use extracted sequences of graphs and patterns to model, Rioult et al.~\cite{rioult2014mining} and Drachen et al.~\cite{drachen2014skill} use historical data of the teams to model. These works mainly focus on using the model to predict the outcome of the competition~\cite{mora2018moba}. On the other hand, we focus on the relationship between the overall surprise and the reward of the ``Game Changer''. Therefore, we provide a more abstract model to describe the game.

\section{Problem Statement: Quidditch}\label{sec:prob}
This section states the formal optimization goal and transition process. We first focus on Quidditch and will extend the setting to general MOBA games in \Cref{sec:moba}. 

We model Quidditch as a multi-round game between two teams, say Gryffindor and Slytherin. In each round, with probability $q$, one team will catch the Golden Snitch (worth $x\in\mathbb{N}$ points), and with probability $1-q$, one team will score $1$ point. The game ends immediately when the Golden Snitch is caught, and the team with the higher score wins the game. To break a tie, the team that catches the Golden Snitch wins if two teams have the same score. We assume that each team's scoring probability across rounds is a constant, i.e., Gryffindor scores with the same probability $p$ in each round. \Cref{fig:trans} shows the transitions when the Snitch's score is 0 and 1.

\subsection{Optimization Goal}

\paragraph{Belief curve}
The belief curve is a random-length sequence of random variables {\normalsize\[ \belset:=(\bel_0,\bel_1,\ldots,\bel_n), \]}
where $\bel_i\in[0,1]$ is the belief for the probability that Gryffindor wins the whole game after round $i$. $\bel_0$ is the initial belief. $n$ is the round in which the Golden Snitch is caught, i.e, the game ends. Note that $n$ is also a random variable. $\bel_n$ is either zero or one since the outcome must be revealed when the game ends.

\begin{definition}[Surprise~\cite{ely2015suspense}] 
Given a belief curve $\belset$, the \emph{amount of surprise generated by round $i$} is defined as $\Delta_\belset^i:=|\bel_i-\bel_{i-1}|$, and the \emph{overall surprise} for a given belief curve is \[\Delta_\belset:= \sum_i \Delta_\belset^i.\]
\end{definition}
\paragraph{Maximizing the expected overall surprise} 
We aim to compute the score of the Golden Snitch $x$ which, in expectation, maximizes the overall surprise. That is, we aim to find the optimal $x^*=\argmax_{x}\ \E[\Delta_\belset(x)]$, where $\Delta_\belset(x)$ is the overall surprise where the score of the Golden Snitch is equal to $x$. 

\subsection{Time-homogeneous Markov Chain}\label{sec:markov}
To calculate the belief and surprise, we restate the process of Quidditch as a Markov chain.

\paragraph{State space}
Define the state space $\mathbb{S}$ and a sequence of random variables $\mathcal{D}^{(i)}$ as follow:
\begin{align*}
\mathbb{S}=&\left\{
\begin{array}{rcl}
\mathbb{Z},&\text{Gryffindor's score minus Slytherin's score}\\
\lose,&\text{Slytherin wins the whole game}\\
\win,&\text{Gryffindor wins the whole game}
\end{array}
\right\}\\
D_i\in& \mathbb{S},\notag\\
\mathcal{D}^{(i)}=&(D_0,D_1,\ldots,D_i),\qquad D_0,\ldots,D_i\in\mathbb{Z}\\
\mathcal{D}=&(D_0,D_1,\ldots,D_n),\qquad D_0,\ldots,D_{n-1}\in \mathbb{Z}, D_n\in\{\lose,\win\}
\end{align*}

$D_i$ is the state in round $i$, $\mathcal{D}^{(i)}$ is the history of first $i$ rounds and $\mathcal{D}$ is the process of the whole game. Random variable $n$ is the total number of rounds. Since we assume the outcome of these rounds are independent, the value of $D_i$ only depends on $D_{i-1}$. Thus, $\mathcal{D}$ is a Markov Chain.

\begin{figure}[t]\centering
\begin{minipage}{0.3\linewidth}\centering
    \includegraphics[width=\linewidth]{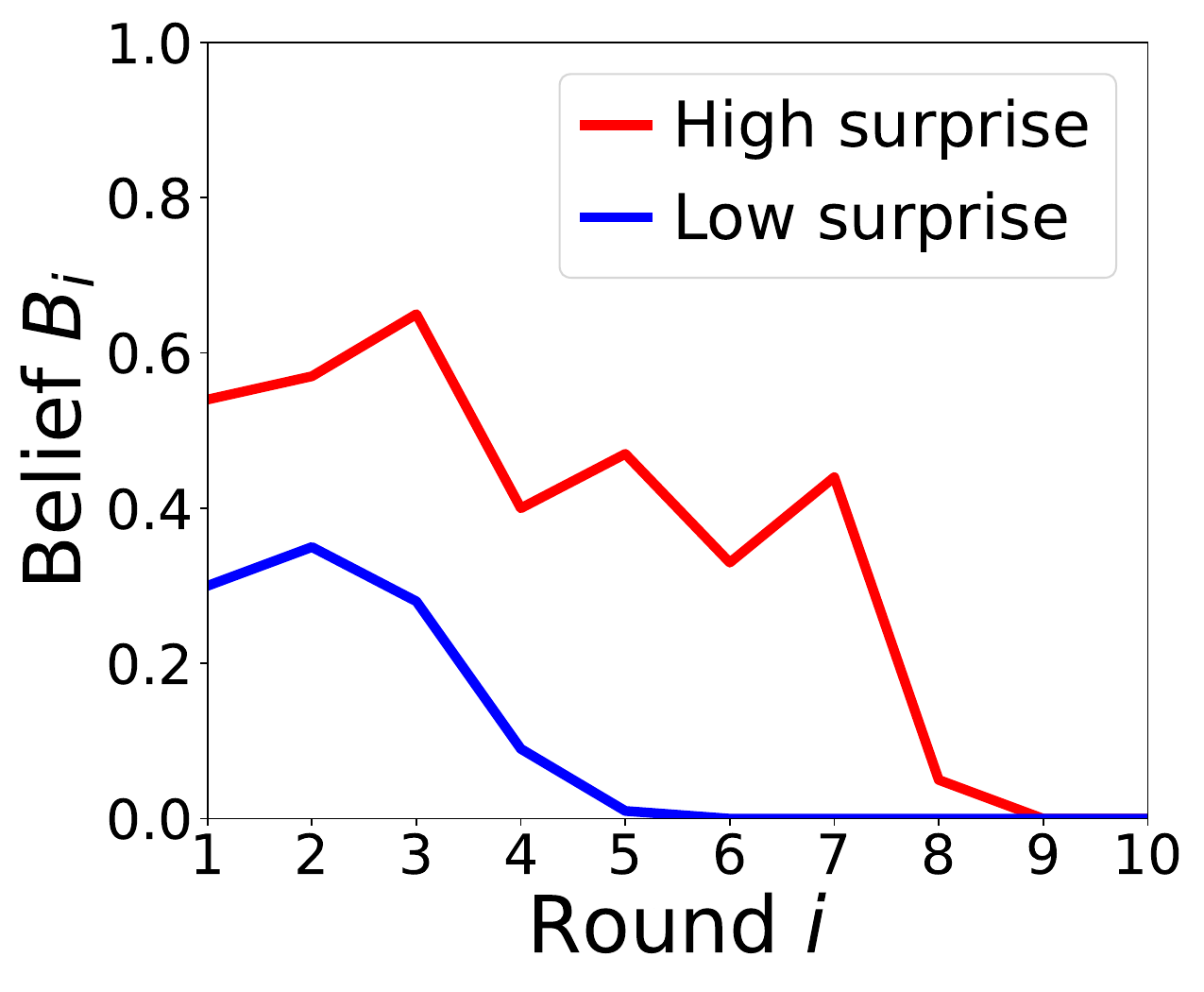}
    \caption{\textbf{Belief curves with low/high overall surprise}}
    \label{fig:examples}
\end{minipage}
\hfill
\begin{minipage}{0.65\linewidth}\centering
    \includegraphics[width=.8\linewidth]{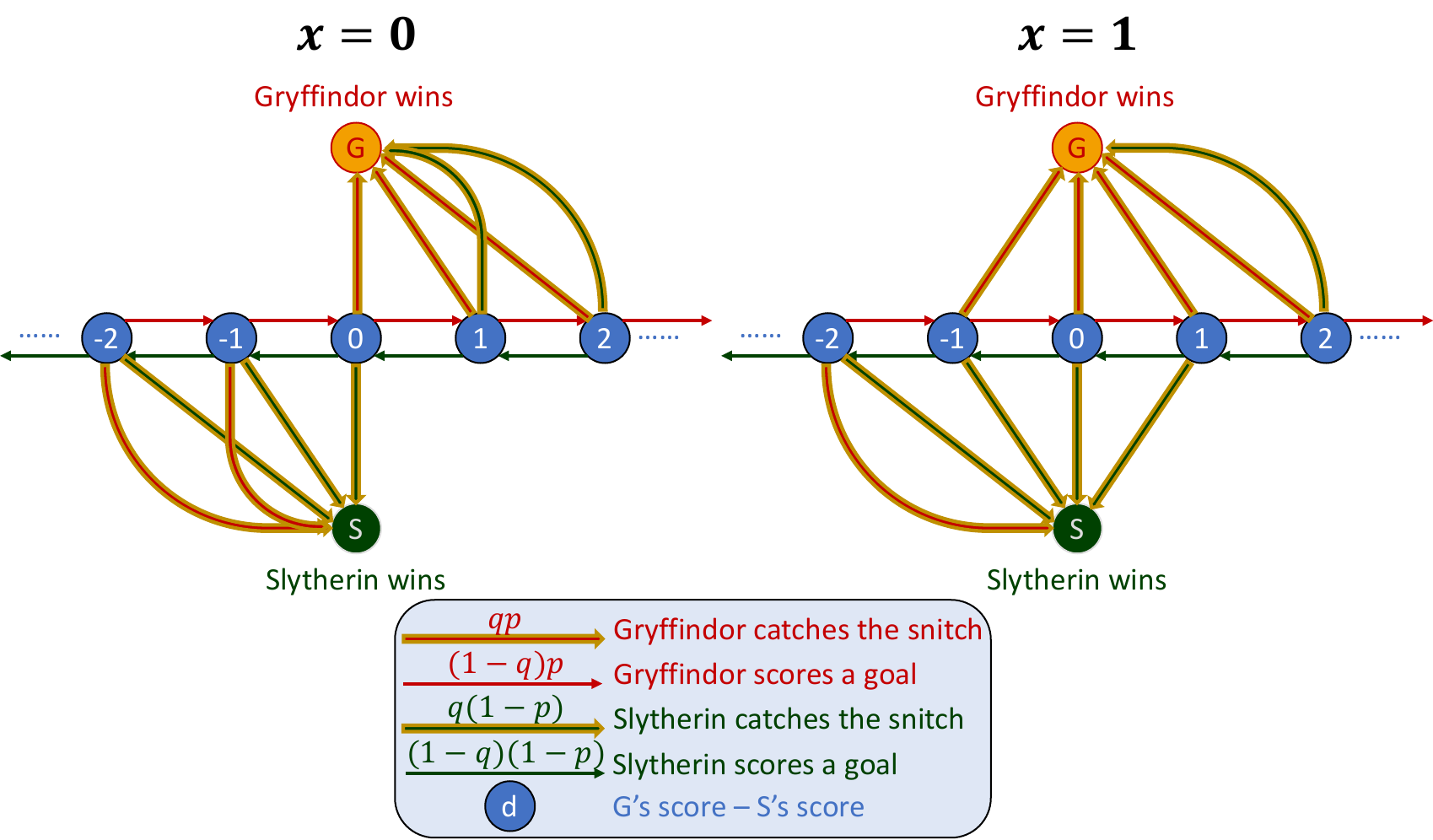}
    \caption{\textbf{Transitions when the Snitch's score is $0$ and $1$}}
    \label{fig:trans}
\end{minipage}
\end{figure}

\paragraph{Transitions} Then we can write the transitions of the Markov chain $\mathcal{D}$: for all $i\geq 0$,
\begin{align}
&\Pr[D_{i+1}=\delta+1|D_i=\delta]=(1-q)\times p\tag{G scores a point},\\
&\Pr[D_{i+1}=\delta-1|D_i=\delta]=(1-q)\times (1-p)\tag{S scores a point},\\
&\Pr[D_{i+1}=\win|D_i=\delta]=q\times 
\begin{cases}
0,&\delta<-x\\
p,&-x\leq \delta\leq x\\
1,&\delta>x\\
\end{cases}\label{fo:trans1},\\
&\Pr[D_{i+1}=\lose|D_i=\delta]=q\times 
\begin{cases}
1,&\delta<-x\\
1-p,&-x\leq \delta\leq x\\
0,&\delta>x\\
\end{cases}\label{fo:trans2},
\end{align}
where formulas~\eqref{fo:trans1}\eqref{fo:trans2} indicate the transitions that occur when a team catches the Golden Snitch. When Gryffindor trails by more than $x$ points, i.e., $\delta<-x$, Slytherin will win no matter who catches the Golden Snitch. When the score difference is between $-x$ and $x$, the team who catches the snitch wins the game. Other cases are similar. 

The probabilities of the transitions are independent of $i$, so $\mathcal{D}$ is a time-homogeneous Markov chain~\cite{serfozo2009basics}. Next, we can define the belief curve $\belset$ formally based on the time-homogeneous Markov chain $\mathcal{D}$.

We define $O$ as the outcome of the whole game, that is
\[\normalsize
O=
\begin{cases}
\lose,&\text{Slytherin wins the whole game}\\
\win,&\text{Gryffindor wins the whole game}
\end{cases}.
\]
In this setting, $\bel_i=\Pr[O=\win|\mathcal{D}^{(i)}]$, i.e., $\bel_i$ is the conditional probability that Gryffindor wins the whole game. Since $\mathcal{D}$ is a Time-homogeneous Markov chain, we immediately obtain the following observation. 

\begin{observation}
The belief only depends on the current score difference between two teams, i.e., there exists a sequence $b_{\delta},\delta\in(-\infty,\infty)$ such that for all history $\mathcal{D}^{(i)}=(D_0,D_1,\ldots,D_i)$, $B_i=\Pr[O=\win|\mathcal{D}^{(i)}=(D_0,D_1,\ldots,D_i)]=b_{D_i}$.
\end{observation}

\paragraph{The recurrence relationship of belief}
Based on the transition probabilities of Markov chain $\mathcal{D}$, the recurrence relationship of belief follows that
\begin{align}
b_\delta=&
\left(1-q\right)\left(pb_{\delta+1}+\left(1-p\right)b_{\delta-1}\right)\notag+q\times 
\begin{cases}
0, & {\delta<-x}\\
p, & {-x \leq \delta \leq x}\\
1, & {\delta > x}\\
\end{cases}.
\end{align}

\subsection{Method Overview}
In this section, we briefly introduce our method to solve the optimization problem.
\paragraph{Step 1 (Solving recurrence relationship of belief)} First, we transform the recurrence relationship into the format of matrix multiplication. Then we use eigendecomposition to find the power of the matrices and get the expression of $b_\delta$ about $b_0, b_1$. Note that when the absolute difference $|\delta|$ goes to infinity, the belief goes to $1$ or $0$, i.e., $\lim_{\delta\rightarrow\infty}b_\delta=1,\lim_{\delta\rightarrow\infty}b_\delta=0$. The limit cases imply the values of $b_0, b_1$. By substituting $b_0,b_1$ into the expression of $b_\delta$, we finally get the general form.

\paragraph{Step 2 (Belief $\rightarrow$ Surprise)} Define $s_\delta$ as a single round expected surprise when the difference is $\delta$. Let $v_\delta$ be the expected number of visits to the state $D_i=\delta$ in the Markov chain $\mathcal{D}$. Then, the expected overall surprise is \[\normalsize\E[\Delta_\belset(x)]=\sum_{\delta}s_\delta\cdot v_\delta.\] For convenience, we calculate the surprise generated by catching Snitch ($s^{\text{final}}_\delta$) and scoring a goal $(s^{\text{non-final}}_\delta)$ separately:
\begin{align*}
s_\delta =& s^{\text{final}}_\delta + s^{\text{non-final}}_\delta\\
s^{\text{non-final}}_\delta=&(1-q)\big(p(b_{\delta+1}-b_\delta)+(1-p)(b_\delta-b_{\delta-1})\big)\\
s^{\text{final}}_\delta=& q\times
\begin{cases}
1-b_\delta, & \delta\geq x\\
p(1-b_\delta) + (1-p)b_\delta, & -x\leq \delta\leq x\\
b_\delta, & \delta\leq -x
\end{cases}.
\end{align*}
Based on the transitions of Markov chain $\mathcal{D}$, the recurrence relationship of $v_\delta$ is 
{\normalsize\begin{align*}
v_\delta=
\left(1-q\right)\left(pv_{\delta-1}+(1-p)v_{\delta+1}\right)+q\times 
\begin{cases}
0, & {\delta\neq 0}\\
1, & {\delta=0}
\end{cases}.
\end{align*}}

By transforming the recurrence to matrix power, we can use eigendecomposition to obtain the general formula. By substituting the formulas of $s_\delta$ and $v_\delta$ into the overall expected surprise, we obtain a format of the summation of a geometric sequence, which leads to a closed-form formula for the overall surprise. 

\paragraph{Step 3 (Optimal score)} We define $\surp(x)=\E[\Delta_\belset(x)],x\geq 0$ as the continuous generalization of the expected overall surprise. By bounding the first and second derivative of $\surp(x)$, we find the upper bound of the optimal score. numerical studies show that the true optimal solution is either zero or very close to the upper bound.

\newcommand{\numerator}{}
\newcommand{\denominator}{\mathbb{N}}

\section{Results of Quidditch}\label{sec:quidditch}
\begin{restatable}[Main theorem]{theorem}{thmmain}
\label{thm:main}
In Quidditch, there exists a closed-form formula for the overall expected surprise. For all $0<p\leq \frac{1}{2}$, $q\in(0,1)$\footnote{This assumption does not lose generality since we can exchange two teams.}, 
\begin{itemize}
    \item \textbf{Symmetric $p=\frac{1}{2}$:} the optimal score is zero, i.e.,  $x^*=0$;
    \item \textbf{General $p<\frac{1}{2}$:} the optimal score $x^*\leq \lceil U(p,q)\rceil$, where $U(p,q)=\max\left\{1,-\frac{C_1}{C_2}-\frac{1}{\log(\B)}\right\}$\footnote{$C_1,C_2,\B$ are closed form functions of $p,q$: $\kappa = \sqrt{1-4(1-p)p(1-q)^2},\beta=\frac{1+\kappa}{2p(1-q)}, \hat{\beta}=\frac{1-\kappa}{2(1-p)(1-q)}$, $C_1= \frac{\left(1 - q\right) \left(1 - \B\right) \left(1 - \C\right) \left(2 p + \left(1 - p\right) \left(1/\B + \C\right)\right)}{\left(1 - \BC\right)^2} - \frac{\left(\left(1 - 2 p\right) q \left(1 - \C\right) - \C q \left(1 - \B\right)\right) }{ \left(1 - \BC\right)^2} - \frac{2 q \left(1 - p\right)}{1 - \B},C_2 =  \frac{\left(1 - q\right)  \left(1 - \B\right)  \left(1 - \C\right) \left(p + \left(1 - p\right) / \B\right)}{\left(1 - \BC\right)} - \frac{q \left(1 - 2 p\right) \left(1 - \C\right)}{\left(1 - \BC\right)}$. }. When $p$ is fixed and $q\rightarrow 0$, $U(p,q)=\frac{1}{2q}\left(\frac{1-p}{p}-1\right)+O(1)$. 
\end{itemize}
\end{restatable}

The proof of the theorem is divided into three steps. The first step is to find the \textbf{belief} function of the audience in each state. The second step is to use the belief function to construct a closed form of \textbf{surprise} for a game. The last step is to find the extrema of surprise.

\subsection{Belief}
The first thing we need is to calculate belief $b_\delta$ for each $\delta$. 

Recall the recurrence relationship in \Cref{sec:markov},
{\normalsize\[
b_\delta=
\left(1-q\right)\left(pb_{\delta+1}+\left(1-p\right)b_{\delta-1}\right)+q\times
\begin{cases}
0, & {\delta<-x}\\
p, & {-x \leq \delta \leq x}\\
1, & {\delta > x}\\
\end{cases}.
\]}

By transforming the formula, we get
{\normalsize\[
b_\delta=
\begin{cases}
\frac{1}{(1-q)p}b_{\delta-1}-\frac{1-p}{p}b_{\delta-2}-\frac{q}{(1-q)p},&\delta>x+1\\
\frac{1}{(1-q)p}b_{\delta-1}-\frac{1-p}{p}b_{\delta-2}-\frac{q}{1-q},&1\leq \delta\leq x+1\\
\frac{1}{(1-q)(1-p)}b_{\delta+1}-\frac{p}{1-p}b_{\delta+2}-\frac{qp}{(1-q)(1-p)},&-x-1\leq \delta\leq 0\\
\frac{1}{(1-q)(1-p)}b_{\delta+1}-\frac{p}{1-p}b_{\delta+2},&\delta< -x-1\\
\end{cases}.
\]}

Let $\alpha=\frac{1+\kappa}{2p(1-q)},\beta=\frac{1-\kappa}{2p(1-q)}, \hat{\alpha}=\frac{1+\kappa}{2(1-p)(1-q)}=\frac{p}{1-p}\alpha,\hat{\beta}=\frac{1-\kappa}{2(1-p)(1-q)}=\frac{p}{1-p}\beta$. By calculation the power of transition matrix, we get the general formula for $b_\delta$.

{\normalsize\begin{align*}
b_\delta = &
\begin{cases}
\frac{-\beta^{\delta-x}((\alpha-1)(1-p)+\beta^x (b_1-\alpha b_0+(\alpha-1)p))}{\alpha-\beta}+1, &\delta\geq x\\
\frac{\alpha^\delta(b_1-\beta b_0+(\beta-1)p)-\beta^\delta(b_1-\alpha b_0+(\alpha-1)p)}{\alpha-\beta}+p, &0\leq \delta\leq x\\
\frac{\hat{\alpha}^{\delta-1}(b_0-\hat{\beta} b_1+(\beta-1)p)-\hat{\beta}^{\delta-1}(b_0-\hat{\alpha} b_1+(\hat{\alpha}-1)p)}{\hat{\alpha}-\hat{\beta}}+p, &-x\leq \delta<0\\
\frac{-\hat{\beta}^{-\delta-x}(-(\hat{\alpha}-1)p+\hat{\beta}^{x+1}(b_0-\hat{\alpha} b_1+(\hat{\alpha}-1)p))}{\hat{\alpha}-\hat{\beta}},&\delta\leq -x
\end{cases}
\end{align*}}
Using the fact that $\alpha=\frac{1}{\hat{\beta}},\hat{\alpha}=\frac{1}{\beta}$, we eliminate $\alpha, \hat{\alpha},b_0,b_1$ and obtain the following succinct form:

{\normalsize\[
b_\delta=
\begin{cases}
\frac{-\beta^{\delta-x}(1-\hat{\beta})(1-p+\beta^{2x+1}p)}{1-\beta\hat{\beta}}+1, &\delta\geq x\\
\frac{\hat{\beta}^{x-\delta+1}(1-\beta)(1-p)-\beta^{x+\delta+1}(1-\hat{\beta})p}{1-\beta\hat{\beta}}+p, &-x\leq \delta\leq x\\
\frac{\hat{\beta}^{-\delta-x}(1-\beta)(p+(1-p)\hat{\beta}^{2x+1})}{1-\beta\hat{\beta}},&\delta\leq -x
\end{cases}.
\]}

\subsection{Expected Number of Visits}
Next, we consider the expected number of visits to each state. Define $v_\delta$ as the expected number of visits when the state is $\delta$. 
{\normalsize\[
v_\delta=(1-q)(pv_{\delta-1}+(1-p)v_{\delta+1})+q\times
\begin{cases}
0,&\delta\neq 0\\
1,&\delta=0
\end{cases}
\]}

By simple transformation, we get the recurrence formula
{\normalsize\[
v_\delta=
\begin{cases}
\frac{\frac{v_{\delta-1}}{1-q}-pv_{\delta-2}}{1-p},&\delta>1\\
\frac{\frac{v_{\delta+1}}{1-q}-(1-p)v_{\delta+2}}p,&\delta<-1\\
\end{cases}.
\]}

Let $\kappa = \sqrt{1-4(1-p)p(1-q)^2}$, recall $\alpha=\frac{1+\kappa}{2p(1-q)}$, $\beta=\frac{1-\kappa}{2p(1-q)}$, $ \hat{\alpha}=\frac{1+\kappa}{2(1-p)(1-q)}=\frac{p}{1-p}\alpha$, and $\hat{\beta}=\frac{1-\kappa}{2(1-p)(1-q)}=\frac{p}{1-p}\beta$.
Then we can get the general formula for $v_\delta$:
\[\normalsize
v_\delta=
\begin{cases}
\frac{\hat{\beta}^\delta}{\kappa},&\delta\geq 0\\
\frac{\beta^{-\delta}}{\kappa},&\delta<0
\end{cases}
.
\]

\subsection{Surprise}

We use $s_\delta$ to represent the expectation of surprise generated from a single round starting at state $\delta$. 
Recall Figure~\ref{fig:trans}, each state has two types of transitions (Type 1: G/S catches the Snitch, Type 2: G/S scores a goal) to generate surprise. We divide it into two categories $s^{\text{final}}_\delta$ and $s^{\text{non-final}}_\delta$. That is, $s^{\text{final}}_\delta$ represents the expected surprise generated by transitions of Type 1, and $s^{\text{non-final}}_\delta$ represents the expected surprise generated by transitions of Type 2. Formally,

\begin{align}
s_\delta = & s^{\text{final}}_\delta + s^{\text{non-final}}_\delta\notag\\
s^{\text{non-final}}_\delta = & \left(1-q\right)\times\left(\left(b_{\delta+1}-b_\delta\right)p+\left(b_\delta-b_{\delta-1}\right)\left(1-p\right)\right)\label{fo:beliefdiff}\\
s^{\text{final}}_\delta = &
\begin{cases}
\left(1-b_\delta\right)q, & \delta\geq x\\
\left(\left(1-b_\delta\right)p + b_\delta\left(1-p\right)\right) q, & -x\leq \delta\leq x\\
b_\delta q, & \delta\leq -x
\end{cases}.
\end{align}

We first calculate the difference of $b_\delta$.

{\normalsize\[
b_{\delta+1}-b_\delta=
\begin{cases}
\frac{\beta^{\delta-x}(1-\beta)(1-\hat{\beta})(1-p+\beta^{2x+1}p)}{1-\beta\hat{\beta}}, &\delta\geq x\\
\frac{(\frac1{\hat{\beta}}-1)\hat{\beta}^{x-\delta+1}(1-\beta)(1-p)-(\beta-1)\beta^{x+\delta+1}(1-\hat{\beta})p}{1-\beta\hat{\beta}}, &-x\leq \delta\leq x\\
\frac{\hat{\beta}^{-\delta-x-1}(1-\beta)(1-\hat{\beta})(p+(1-p)\hat{\beta}^{2x+1})}{1-\beta\hat{\beta}},&\delta< -x
\end{cases}.
\]}

Recall the formula for $\E[\Delta_\belset(x)]$
\begin{align*}
\E[\Delta_\belset(x)]=&\sum_{\delta}s_{\delta}\cdot v_\delta=\underbrace{\sum_{\delta}s^{\text{non-final}}_\delta\cdot v_\delta}_{\textbf{Non-Final}}+\underbrace{\sum_{\delta}s^{\text{final}}_\delta\cdot v_\delta}_{\textbf{Final}}.
\end{align*}

We sum the two parts to get

{\footnotesize\begin{align*}
&\E[\Delta_\belset(x)]\\
=&\sum_{\delta}v_\delta\cdot s_\delta\\
=&\frac{(1-q)(1-\beta)(1-\hat{\beta})}{\kappa(1-\beta\hat{\beta})}\\
&\times\Bigg((p+(1-p)\hat{\beta})\bigg(\frac{(1-p)\hat{\beta}^{x}+p\beta^{x+1}}{1-\beta\hat{\beta}}+(1-p)x\hat{\beta}^x\bigg)\\
&+(p\beta+(1-p))\bigg(\frac{p\beta^{x}+(1-p)\hat{\beta}^{x+1}}{1-\beta\hat{\beta}}+p x\beta^{x}\bigg)\Bigg)\\
&+\frac{q}{\kappa(1-\beta\hat{\beta})}\Bigg(\frac{\beta(1-\hat{\beta})(1-p+\beta^{2x+1}p)\hat{\beta}^{x+1}}{1-\beta\hat{\beta}}+2(1-p)p\frac{(1-\hat{\beta}^{x+1})(1-\beta\hat{\beta})}{1-\hat{\beta}}\\
&+(1-2p)\left((1-\beta)(1-p)x\hat{\beta}^{x+1}-(1-\hat{\beta})p\beta^{x+1}\frac{1-(\beta\hat{\beta})^{x+1}}{1-\beta\hat{\beta}}\right)\\
&+(1-2p)\left(-(1-\hat{\beta})p x\beta^{x+1}+(1-\beta)(1-p)\hat{\beta}^{x+1}\frac{1-(\beta\hat{\beta})^{x+1}}{1-\beta\hat{\beta}}\right)\\
&+2(1-p)p\frac{(\beta-\beta^{x+1})(1-\beta\hat{\beta})}{1-\beta}+\frac{\hat{\beta}(1-\beta)(p+(1-p)\hat{\beta}^{2x+1})\beta^{x+1}}{1-\beta\hat{\beta}}\Bigg).
\end{align*}}
We start to prove Theorem~\ref{thm:main} by deriving the closed-form formula for the overall expected surprise. 

\subsection{Optimal \(x^*\)}
We define the continuous generalization of the overall expected surprise as $\surp(x)=\E[\Delta_\belset(x)]$ whose domain is $x\in[0,+\infty)$. We write $\surp(x)$ in the form where $x$ is the principal element: $\surp(x)=C_0 + C_1 \B^x + C_2 x \B^x + C_3 \B^{2x}\C^x + \hat{C}_1 \C^x + \hat{C}_2 x \C^x + \hat{C}_3 \B^x\C^{2x}$, where each $C_i,\hat{C}_i$ and $\beta,\hat{\beta}$ are functions of $p$ and $q$. 

Here are several properties of these coefficients that will be used later. The first two properties follow from the definition and the third is proved by Mathematica. 
\begin{claim}
When $0<p\leq\frac{1}{2},0<q<1$,
\begin{enumerate}
    \item $\forall i\in\{1,2,3\}, C_i(p,q)=\hat{C}_i(1-p,q)$,
    \item $\B(p,q)=\C(1-p,q)$,
    \item $C_0>0$, $C_2,C_3,\hat{C}_1,\hat{C}_2,\hat{C}_3>0$, $0<\B\leq\C<1$.
\end{enumerate}
\label{attributes}
\end{claim}

\paragraph{Proof outline} We upper-bound the derivative of $\surp(x)$ by a relatively simple format and then shows that the upper-bound function must be strictly less than zero when $x\geq \max\{U(p,q),1\}$. Thus, the local maximal/minimal of $\surp(x)$ is less than $\max\{U(p,q),1\}$ which proves our upper-bound result. For even case, we show that $\max\{U(p,q),1\}$ must be less than $1$. Thus, by directly comparing $\surp(0)$ and $\surp(1)$, we show that $\surp(0)$ is a global maximum. 

We start by calculating the derivative of $\surp(x)$. 
\begin{align*}
\frac{\mathrm{d}}{\mathrm{d}x}\surp(x)= & \left(C_1\log\B+C_2+C_2\log\B x\right)\B^x+ C_3\log\left(\B^2\C\right)\B^{2x}\C^x\\
+ & \left(\hat{C}_1\log\B+\hat{C}_2+\hat{C}_2\log\C x\right)\C^x + \hat{C}_3\log\left(\B\C^2\right)\B^x\C^{2x}.
\end{align*}
Note that both $C_3\log\left(\B^2\C\right)\B^{2x}\C^x$ and $\hat{C}_3\log\left(\B\C^2\right)\B^x\C^{2x}$ are negative. We can upper-bound the derivative by function $R(x)$, where
\begin{align*}
R(x) = & \left(C_1\log\B+C_2+C_2\log\B x\right)\B^x + \left(\hat{C}_1\log\C+\hat{C}_2+\hat{C}_2\log\C x\right)\C^x.
\end{align*}
Recall that $\frac{\B}{\C}=\frac{1-p}{p}$, The equation that $R(x)= 0$ is equivalent to that
\begin{align*}
& -\left(C_1\log\B+C_2+C_2\log\B  x\right) = \left(\hat{C}_1\log\C+\hat{C}_2+\hat{C}_2\log\C x\right)\left(\frac{p}{1-p}\right)^x.
\end{align*}
Let $f(x)=C_1\log\B+C_2+C_2\log\B x$ and $g(x)=\hat{C}_1\log\C +\hat{C}_2+\hat{C}_2\log\C x$, since $\left(\frac{p}{1-p}\right)^x$ is always positive, we have $f(x)g(x)<0$ or both $f(x)$ and $g(x)$ are equal to $0$ in order to satisfy the equation. As $f(x)$ and $g(x)$ are both decreasing linear functions, $f(x)g(x)<0$ holds only if $x$ lies in the interval constructed by roots of $
f(x)=0$ and $g(x)=0$, which is $\theta_1=-\frac{C_1}{C_2}-\frac{1}{\log\left(\B\right)}$ and $\theta_2=\frac{\hat{C}_1}{\hat{C}_2}-\frac{1}{\log\left(\C\right)}$. That is, 
\[R(x)=0\Rightarrow
\begin{cases}
x\in\left(\min\left\{\theta_1,\theta_2\right\},\max\left\{\theta_1,\theta_2\right\}\right) & \theta_1\neq\theta_2\\
x = \theta_1 & \theta_1 = \theta_2
\end{cases}.\]

For $\theta_2$, we employ Mathematica to prove the following claim:
\begin{claim}
For all $0<p\leq\frac{1}{2}$, $0<q<1$, $\theta_2(p,q)\leq 0.5$.
\label{asymptotes}
\end{claim}

If $\theta_1\leq\theta_2$, according to Claim \ref{asymptotes}, $1$ is an upper bound of the roots of $R(x)=0$. Otherwise, $\theta_1$ is an upper bound of the roots of $R(x)=0$. Let $U(p,q)=\max\{1,\theta_1\}$. $U(p,q)$ is an upper bound of the roots of $R(x)=0$. When $x>U(p,q)\geq \max\left\{\theta_1,\theta_2\right\}$, $f(x)<0$ and $g(x)<0$, thus $R(x)$ is always negative. Therefore, when $x>U(p,q)$, the $\frac{\mathrm{d}}{\mathrm{d}x}\surp(x)<R(x)<0$. This implies that
\[U(p,q)\geq\argmax_{x\geq 0}\surp(x).\] Since the actual domain of $x$ is $\mathbb{N}$, $\lceil U(p,q)\rceil$ is an upper bound of $x^*$.

For the even case where $p=\frac{1}{2}$, based on Claim \ref{attributes}, we have $C_1=\hat{C}_1,C_2=\hat{C}_2,C_3=\hat{C}_3,\B=\C,\theta_1=\theta_2$. Thus, $U(\frac{1}{2},q)=\max\{1,\theta_1\}\leq 1$, which implies that either $x=0$ or $x=1$ can be optimal.
\begin{align*}
& \surp(0)-\surp(1) = \left(C_0+2(C_1+C_3)\right)-\left(C_0+2(C_1\beta+C_2\beta+C_3\beta^3)\right)\\
> & 2(C_1-(C_1+C_2)\B) = -\frac{q(q^2-2q-1+2\sqrt{q(2-q)})}{2\sqrt{q(2-q)}(q-2+\sqrt{q(2-q)})^2}.
\end{align*}
\begin{claim}\label{claim:surp01}
    $-\frac{q(q^2-2q-1+2\sqrt{q(2-q)})}{2\sqrt{q(2-q)}(q-2+\sqrt{q(2-q)})^2}>0$ if $0<q<1$.
\end{claim}
\begin{proof}[Proof of Claim~\ref{claim:surp01}]
When $0<q<1$,
\begin{align*}
\frac{\mathrm{d}}{\mathrm{d}q}\left(q^2-2q-1+2\sqrt{q(2-q)}\right)=(2-2q)\left(\frac{1}{\sqrt{(2-q)q}}-1\right)> 0.
\end{align*}
So $q^2-2q-1+2\sqrt{q(2-q)}$ is strictly increasing in $q\in (0,1)$. As the value of $q^2-2q-1+2\sqrt{q(2-q)}$ when $q=1$ is $0$, $q^2-2q-1+2\sqrt{q(2-q)}<0$ when $0<q<1$. In addition, $
2\sqrt{q(2-q)}(q-2+\sqrt{q(2-q)})^2>0$ when $0<q<1$. As a result, when $0<q<1$, $-\frac{q(q^2-2q-1+2\sqrt{q(2-q)})}{2\sqrt{q(2-q)}(q-2+\sqrt{q(2-q)})^2}>0$.
\end{proof}
Therefore, $\surp(0)> \surp(1)$. Thus, $x^*=0$ in the even case. 

\subsection{Approximating the Upper-Bound}
In this section, we approximate the upper bound when $q$ goes to zero. We use Taylor expansion to obtain the approximation.

When $q\rightarrow 0^+$, by fixing $0<p<\frac12$ as a constant, we use Taylor expansion to obtain $\beta=1-\frac{q}{1-2p}+O(q^2)$, $C_1=-2(1-p)p+O(q)$, $C_2=\frac{4q(1-p)p^2}{(1-2p)^2}+O(q^2)$, and $\frac{1}{\log\beta}=-\frac{1-2p}{q}+O(1)$. Then we get the approximation of $\theta_1$:
\begin{align*}
\theta_1=-\frac{C_1}{C_2}-\frac{1}{\log\beta}=\frac{(1-2p)^2}{2pq}+O(1)+\frac{1-2p}q+O(1)=\frac{1}{2q}\left(\frac{1-p}{p}-1\right)+O(1).
\end{align*}

\paragraph{Estimating optimal score by upper-bound} When we estimate the optimal score, we pick the best over zero and $\lceil U(p,q)\rceil$. Formally, we define \[\tilde{x}:=\argmax_{x=0,1,...,\lceil U(p,q)\rceil}\E[\Delta_\belset(x)].\] This can be easily calculated with the closed-form formula for surprise. We numerically compare $\tilde{x}$ to the true optimal score (Figure~\ref{fig:optimal}, Figure~\ref{fig:xtilde}). In our numerical study, we sample $10^6$ points in $(p,\frac{1}{q})=((0,0.5),[1.1,100])$ and $\tilde{x}$ equals $x^*$ for 99.9997\% points. Thus, in most area, $\tilde{x}$ provides a pretty good estimation for $x^*$.

\begin{figure}[!ht]\centering
  \subfigure[Optimal $x^*$]{\centering\includegraphics[width=.3\linewidth]{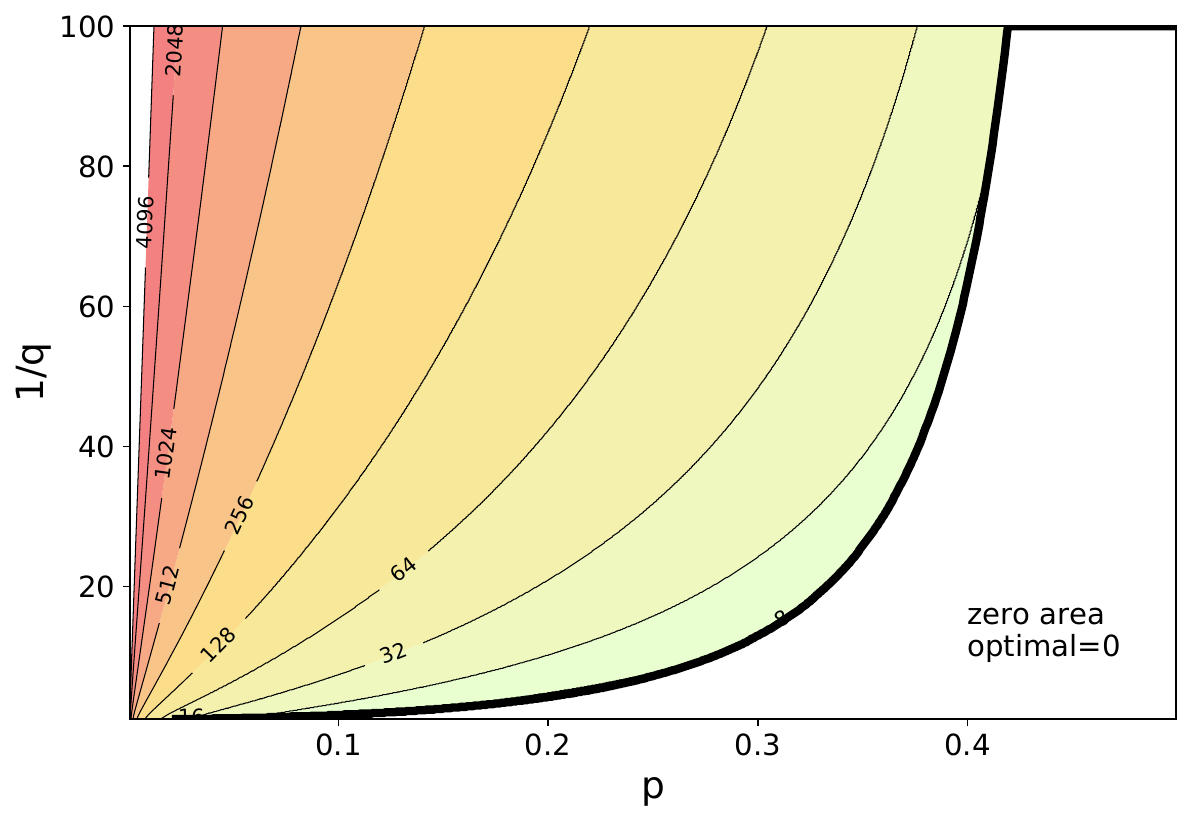}\label{fig:optimal}}
  \subfigure[$\Tilde{x}$]{\centering\includegraphics[width=.3\linewidth]{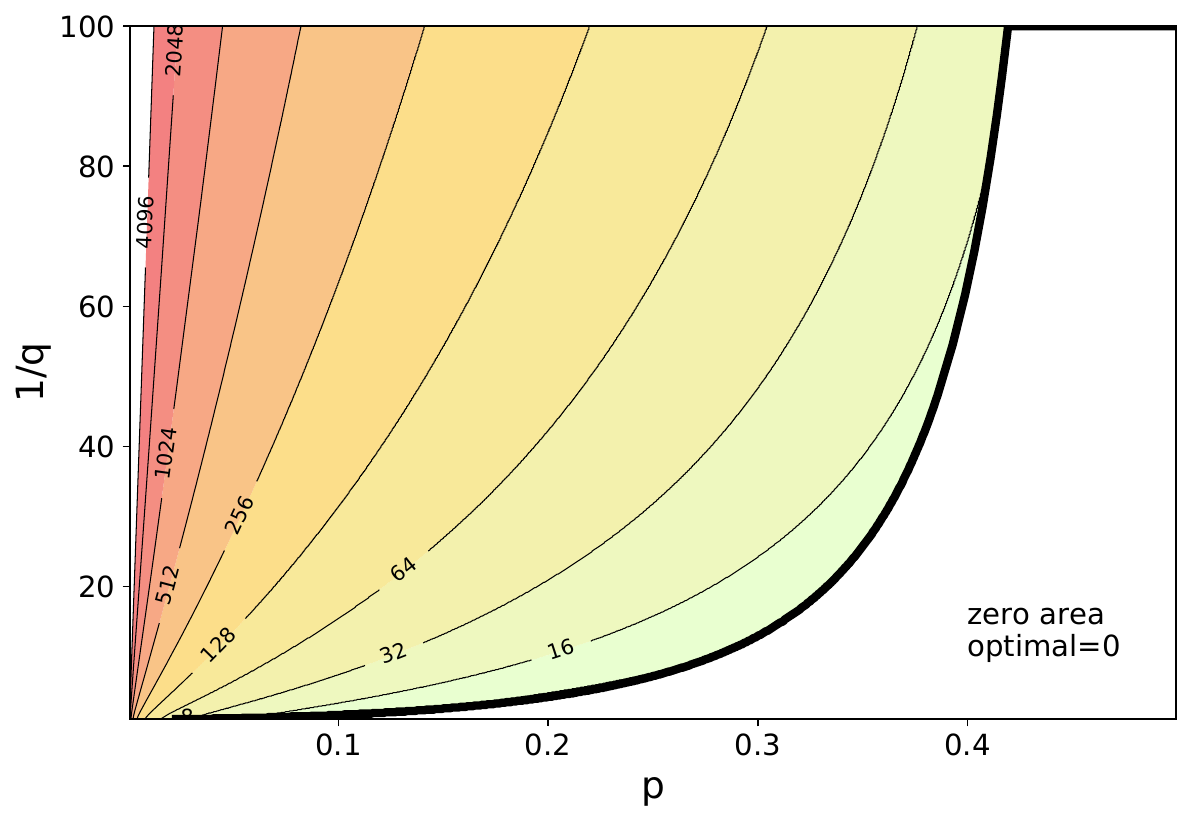}\label{fig:xtilde}}
  \subfigure[$\frac1{2q}(\frac{1-p}{p}-1)$]{\centering\includegraphics[width=.3\linewidth]{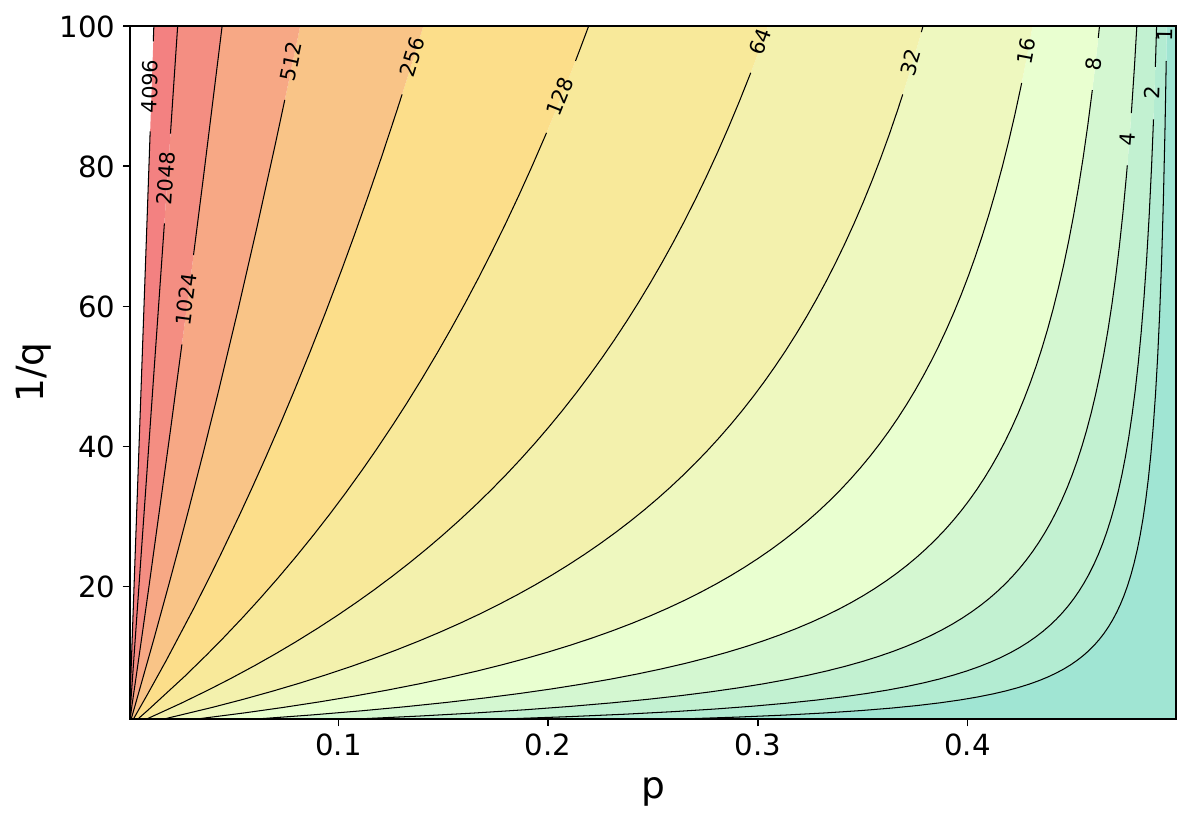}\label{fig:taylor_optimal}}
  \caption{Optimal score and numerical estimations}
\end{figure}

\Cref{fig:optimal} presents the contours of the optimal score for various cases. The result shows that the optimal score increases as the expected time compared to Snitch's difficulty increases ($\frac1q$ increases), and decreases as the match-up becomes more unbalanced ($p$ decreases). There exists an area where the optimal score is zero, called the \emph{zero area}. When $\frac1q$ decreases, the game lasts shorter in expectation, the zero area becomes larger, i.e., the game should be much more uneven in expectation to make the optimal score strictly greater than zero. We also provide numerical results illustrating how the overall surprise depends on the score in \Cref{sec:surp_curve}.

\section{General MOBA Games}\label{sec:moba}

To study the design of ``Game Changer'' in more complicated settings, we design a general model for MOBA games. We then study the effect of the ``Game Changer'' by numerical simulations. We start by introducing our framework model and then fit parameters based on real-world data from LOL and DOTA 2. Finally, we numerically calculate the optimal reward of the ``Game Changer'' for different situations and draw conclusions.

\begin{figure}[!ht]\centering
  \includegraphics[width=.815\linewidth]{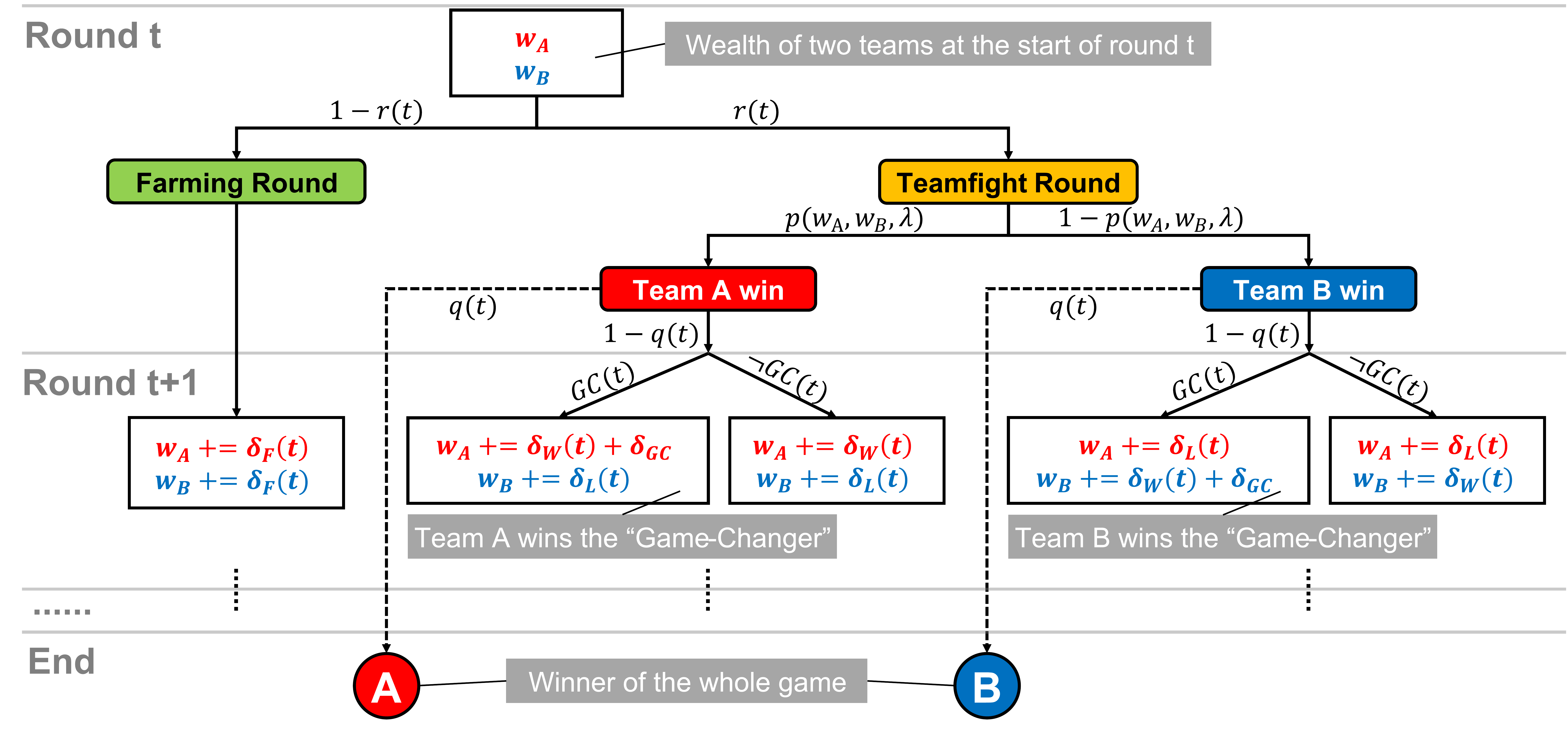}
  \caption{\textbf{Transitions of the MOBA game in one round}}
  \label{fig:mobatrans}
\end{figure}

\subsection{Our Model}\label{model}

MOBA games are real-time strategy games in which two teams (usually 5 players each) fight against each other. Both teams improve their character's strength by earning wealth during the game. The game ends until one team destroys the opponent's main building. The process mainly revolves around character development and cooperative teamfight \cite{yang2014identifying}. We start to describe how our model captures the important features of MOBA games.

\paragraph{Accumulating team wealth}
During the game, players earn wealth to buy items for their characters. The strength of a character can be roughly estimated by the amount of wealth. 

We model MOBA as a multi-round game between two teams. Each round represents one minute of the real-world game. Like Quidditch, the number of rounds is indeterminate. We define the \emph{team wealth} $w_{i}(t),i\in\{A,B\}$ as the sum of the wealth of all characters in team $i$ at the end of round $t$.

At the start of the game, two teams have their basic wealth $w_A(0),w_B(0)$. There are two types of rounds, farming (F) and teamfight (T). Each round $t$ is a teamfight round with independent probability $r(t)$. In a farming round ($type(t)=F$), players earn wealth mainly by destroying buildings and killing periodically respawning creeps. Both teams receive the same team wealth $\deltaw_F(t)$ at the farming round. In a teamfight round ($type(t)=T$), the winning team of the teamfight will get a lot of wealth $\deltaw_W(t)$, while the losing team can only get $\deltaw_L(t)$ which is less than $\deltaw_W(t)$. 

The ``Game Changer'' is spawned in a specific round (like the 20 minutes of the game) and respawned in several rounds after being killed. After a teamfight round, if the ``Game Changer'' exists, the winning team will kill it and gain bonus wealth $\deltaw_{GC}$.

\Cref{teamwealth_def} formally describes how the team wealth $w_{i}(t),i\in\{A,B\}$ are accumulated in different rounds, where \[GC(t)=\begin{cases}\mathrm{False} &\text{the ``Game Changer'' does not exist}\\\mathrm{True} &\text{the ``Game Changer'' exists}\end{cases}.\] For team $i=A,B$,

\begin{equation}
\label{teamwealth_def}
    w_{i}(t)=\begin{cases}
    w_{i}(t-1)+\deltaw_F(t) & type(t)=F\\
    w_{i}(t-1)+\deltaw_L(t) & type(t)=T \land \text{$i$ loses}\\
    w_{i}(t-1)+\deltaw_W(t) & type(t)=T \land \text{$i$ wins} \land \neg GC(t)\\
    w_{i}(t-1)+\deltaw_W(t)+\deltaw_{GC} & type(t)=T \land \text{$i$ wins} \land GC(t)\\
    \end{cases}
\end{equation}

\paragraph{Winning probability of a teamfight $p(w_A(t),w_B(t),\lambda)$}
At round $t$, if it is a teamfight round, the probability that team A wins in this round $p(w_A(t),w_B(t),\lambda)$ depends on their current wealth $w_A(t),w_B(t)$ and the relative ratio of their team ratings $\lambda\in[0,\infty)$. The rating quantifies the team's gaming skills. When two teams have equal ratings, $\lambda=1$. When team A is more skilled, $\lambda>1$.

\paragraph{Endgame}
The end of the game happens when the main building of one team is destroyed by another team. If a teamfight round's winning team kills most characters of another team, before the dead characters revive, the winning team will also have a chance to destroy the main building and win the whole game. The revival time of characters gets longer as the game progresses, which gives more chance for the later teamfight round's winning team to win the whole game. Thus we assume that with probability $q(t)$, a teamfight in round $t$ ends the game immediately and the winner of the teamfight becomes the winner of the game. 

\Cref{fig:mobatrans} summarizes the transitions of the MOBA game in one round.

\subsection{Maximizing the Expected Overall Surprise}

We aim to compute the optimal reward of the ``Game Changer'', $\deltaw_{GC}$ which, in expectation, maximizes the overall surprise. Formally, we aim to find the optimal \[\deltaw_{GC}^*=\argmax_{\deltaw_{GC}} \E[\Delta_\belset(\deltaw_{GC})],\] where $\Delta_\belset(\deltaw_{GC})$ is the overall surprise where the reward of the ``Game Changer'' is $\deltaw_{GC}$. When we calculate the surprise, we pick other functions $r(t),p(t),q(t),\deltaw_W(t),\deltaw_L(t),\deltaw_F(t)$ by fitting the empirical data. We will vary the relative rating ratio $\lambda$ between the two teams. Recall that $\lambda\in [0,+\infty)$ is the relative ratio of the team rating of two teams. Without loss of generality, we consider $\lambda\geq 1$. 

In our model, the process of MOBA is a Markov chain where each state is described by the current time $t$, the current wealth of the two teams $w_A,w_B$ and the time that the last ``Game Changer'' is killed. Given the Markov chain, we calculate the overall expected surprise by backward induction. We also employ approximation techniques based on discretization to reduce the complexity. We defer the detailed description of the algorithm to the full version.

\subsection{Selecting $r(t),p(t),q(t),\deltaw_W(t),\deltaw_L(t),\deltaw_F(t)$}
We use real-world data of the two most popular MOBA games, LOL and DOTA 2, to estimate the teamfight probability function $r(t)$, winning probability function $p(w_A(t),w_B(t),\lambda)$, endgame probability function $q(t)$ and wealth income functions $\deltaw_F(t),\deltaw_W(t),\deltaw_L(t)$.

\paragraph{Game data}

For LOL, we use data of 7046 professional games in a season from LOL World Championship 2020 (exclusive) to LOL World Championship 2021 (inclusive)\footnote{Data is from Games of Legends, https://gol.gg/esports/home/}.
For DOTA 2, we use data of 7230 professional games from The International 2019 (exclusive) to ESL One Stockholm Major 2022 (inclusive)\footnote{Data is from OpenDota API, http://docs.opendota.com/}. We defer the data processing details to the full version.

\paragraph{Teamfight probability $r(t)$}
The solid lines in \Cref{fig:pfight} show the probability of a teamfight in each round $t$ estimated from empirical data. Both the games gradually increase the probability until an upper-bound is achieved. In our numerical studies, we approximate the teamfight probability function by piece-wise linear functions, which are shown by the dashed lines in \Cref{fig:pfight}.

\paragraph{Winning probability of a teamfight $p(w_A,w_B,\lambda)$}

We select the winning probability function from the following family $\{p_{\theta}(\cdot)\}_{\theta}$:

\[
p_{\theta}(w_A,w_B,\lambda)=\begin{cases}
\mathrm{sigmoid}(\theta\frac{\lambda w_A-w_B}{w_B}) & \lambda w_A\geq w_B\\
1-\mathrm{sigmoid}(\theta\frac{w_B-\lambda w_A}{\lambda w_A}) & \lambda w_A < w_B
\end{cases}
\]
where $\mathrm{sigmoid}(x)=\frac{1}{1+\exp(-x)}$.

Since the rating ratio between teams is frequently changing and almost impossible to collect, we consider $\lambda=1$ when fitting the data. We fit the empirical data to select $\theta=9.41$ for LOL and $\theta=5.85$ for DOTA 2. \Cref{fig:winrate} visualizes the empirical winning frequency of a teamfight and our fitted probability.

\begin{figure}[t]\centering
  \subfigure[LOL]{\includegraphics[scale=0.115]{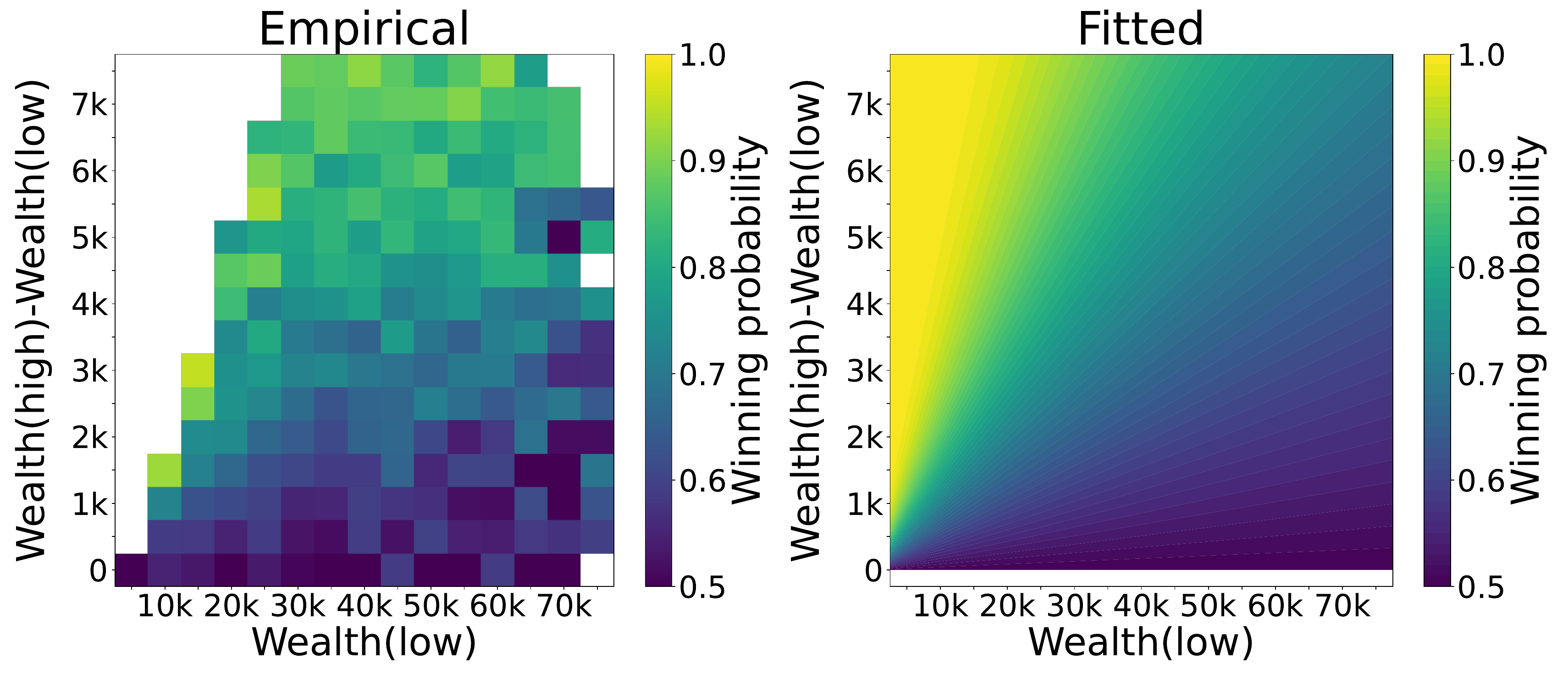}}
  \subfigure[DOTA 2]{\includegraphics[scale=0.115]{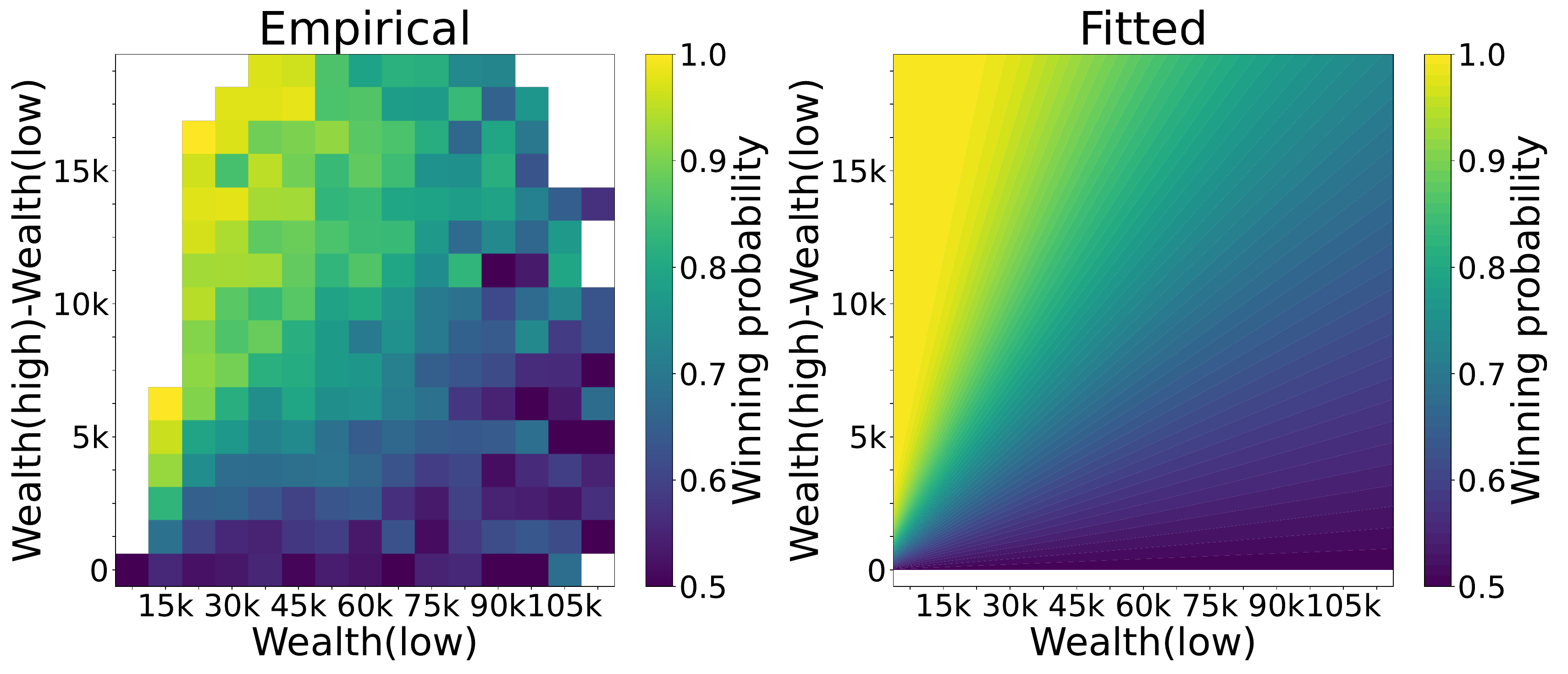}}
  \caption{\textbf{Winning probability of a teamfight when $\lambda=1$}: Left subfigures are the empirical winning frequency estimated from the real-world data. The right subfigures are the winning probability calculated from our model with fitting parameters.}
  \label{fig:winrate}
\end{figure}

\begin{figure}[t]
\centering
\begin{minipage}{0.325\linewidth}
    \centering
    \includegraphics[scale=0.127]{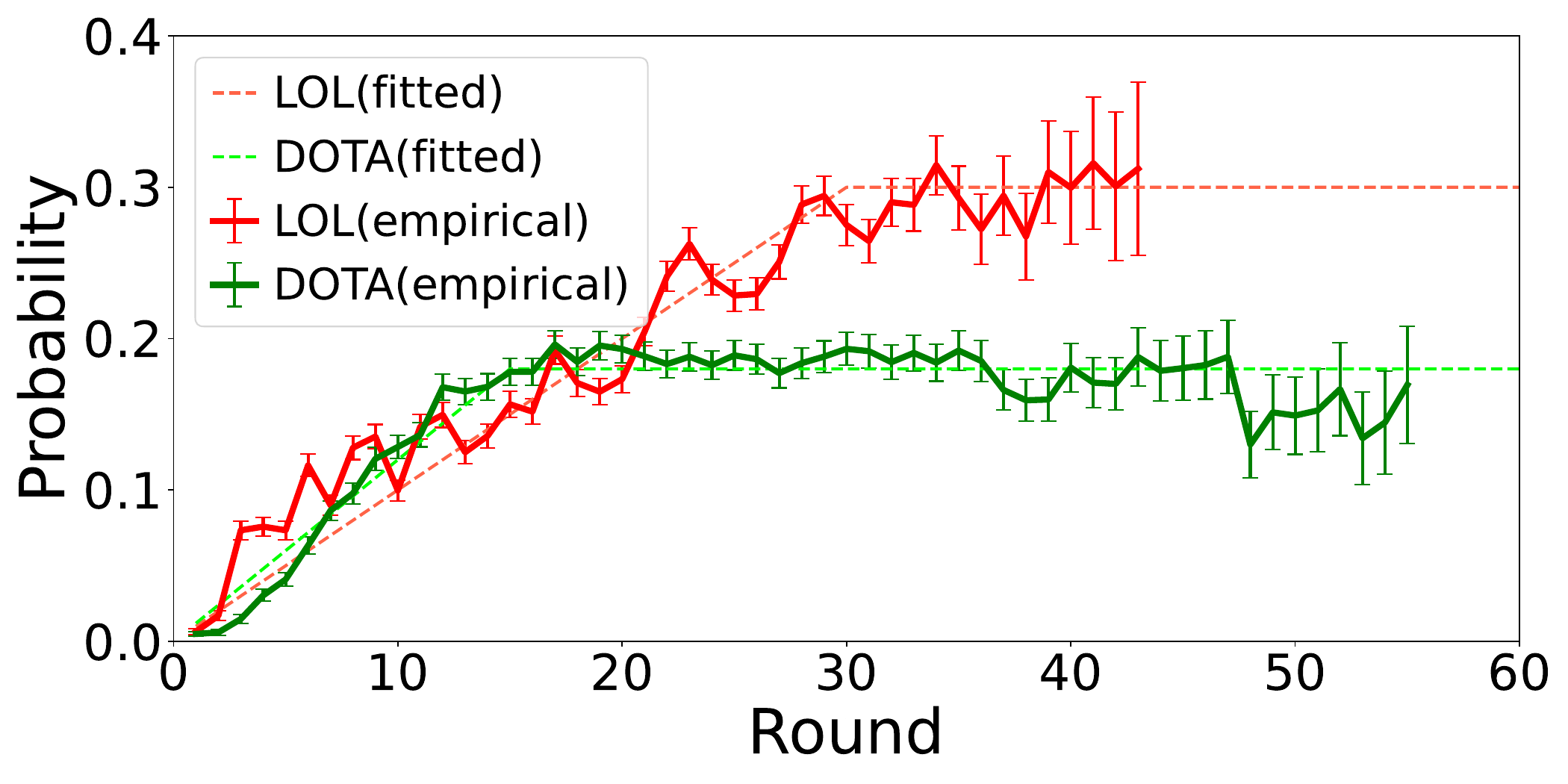}
    \caption{\textbf{Teamfight probability} $r(t)$ in each round}
    \label{fig:pfight}
    \end{minipage}
\hfill
    \begin{minipage}{0.66\linewidth}
    \centering
    \subfigure[LOL]{\includegraphics[scale=0.127]{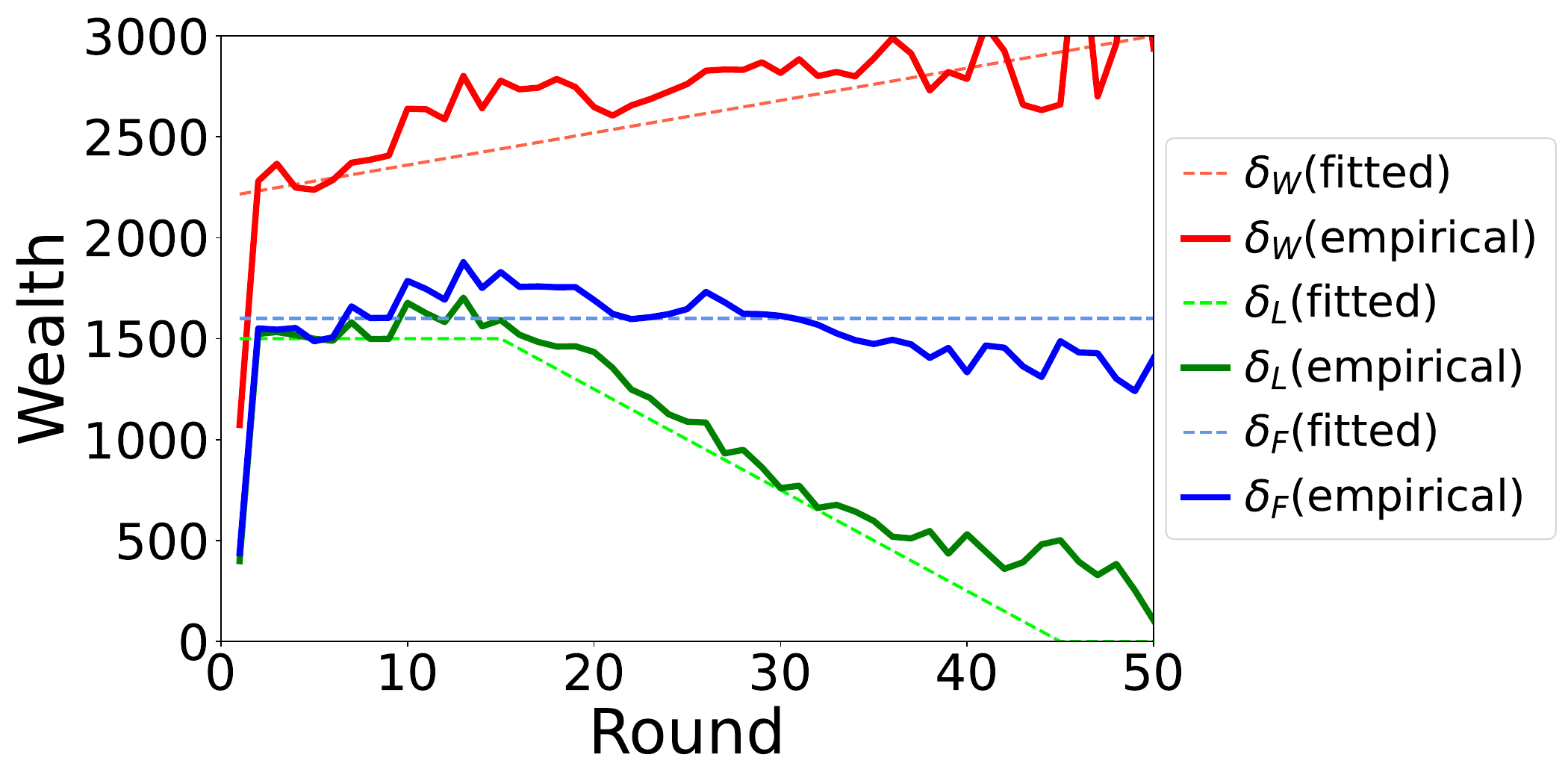}}
    \subfigure[DOTA 2]{\includegraphics[scale=0.127]{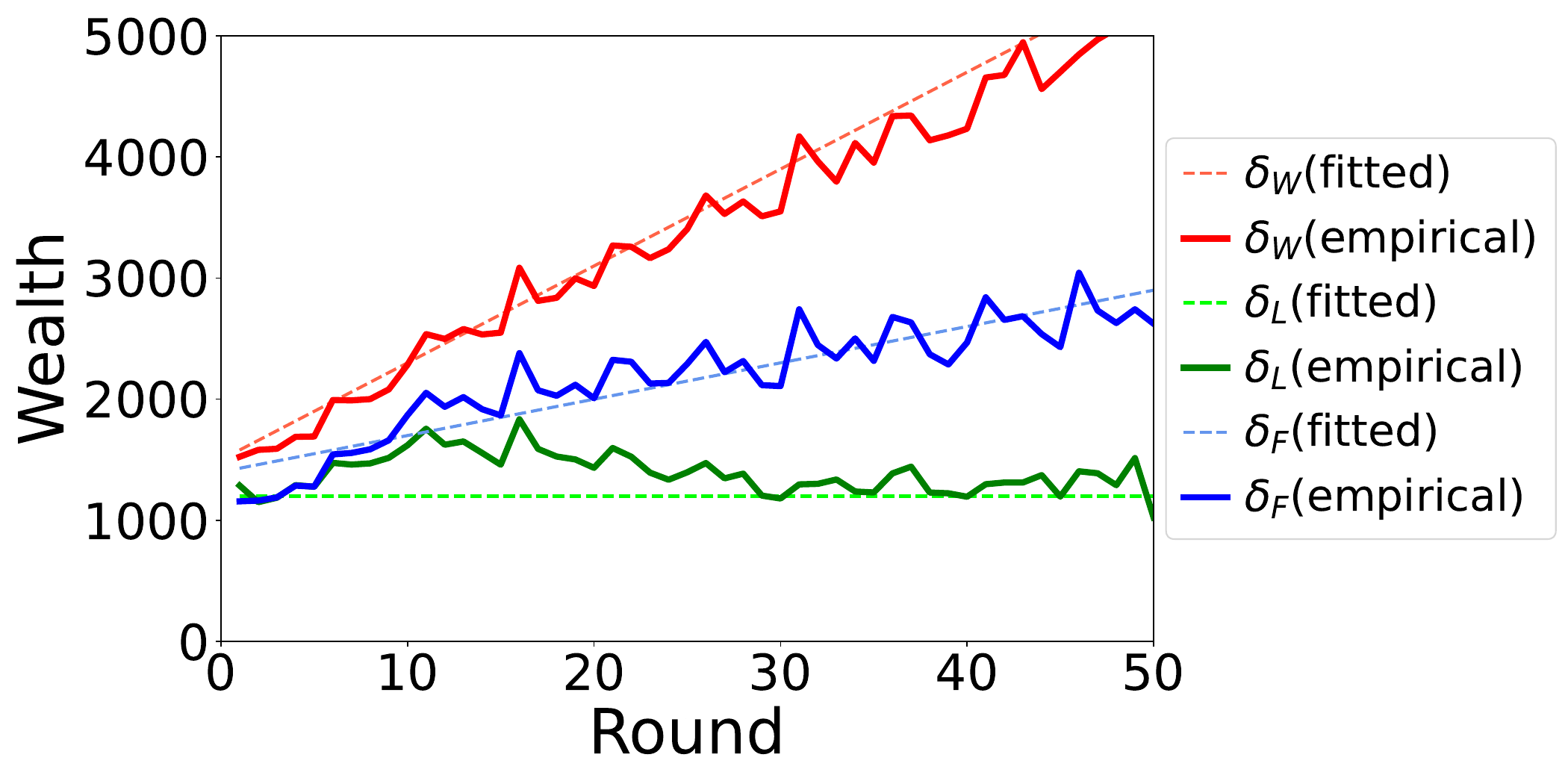}}
    \caption{\textbf{Average wealth income} in each types of round}
    \label{fig:wealth}
    \end{minipage}
\end{figure}

\paragraph{Wealth income per round $\deltaw_F(t),\deltaw_W(t),\deltaw_L(t)$}
The solid lines in \Cref{fig:wealth} show the average wealth gained in each round estimated from empirical data. The blue lines indicate the wealth income $\deltaw_F(t)$ in the farming round. The red lines and green lines represent the wealth income of the winner and loser $\deltaw_W(t),\deltaw_L(t)$ in the teamfight round correspondingly. In our numerical studies, we use piecewise linear functions to approximate them (the dashed lines). 

\paragraph{End probability $q(t)$}
The solid lines in \Cref{fig:pend} show the estimated probability that the a teamfight at round $t$ ends the game conditioning on this round is a teamfight round. The probability function increases regarding $t$. We use piecewise linear functions to approximate them (the dashed lines).

\begin{figure}[t]\centering
\begin{minipage}{0.325\linewidth}
    \centering
    \includegraphics[scale=0.127]{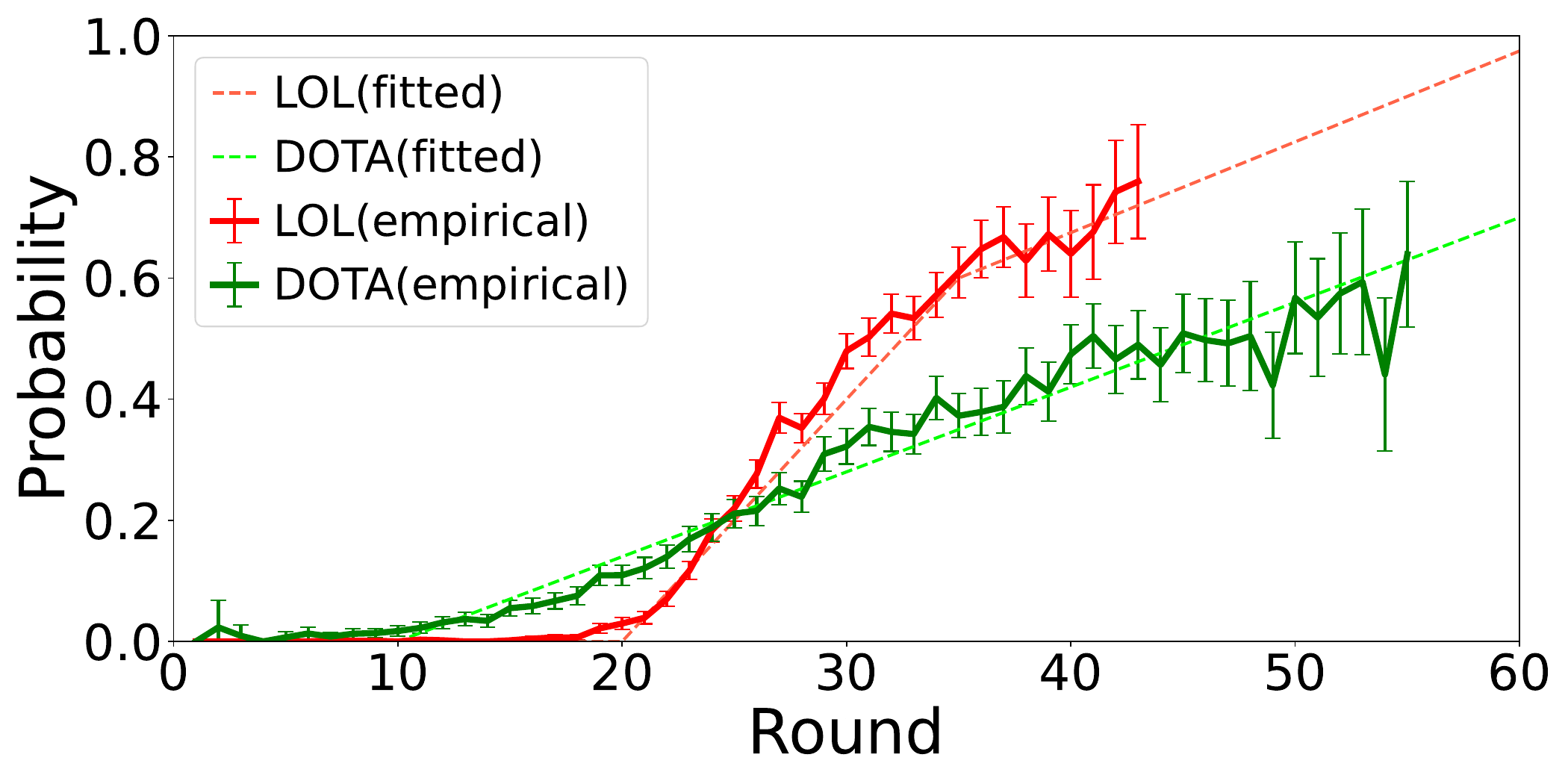}
    \caption{\textbf{Probability $q(t)$ of a teamfight ends the game} in each round $t$}
    \label{fig:pend}
\end{minipage}
\hfill
\begin{minipage}{0.66\linewidth}\centering
    \subfigure[LOL]{\includegraphics[scale=0.127]{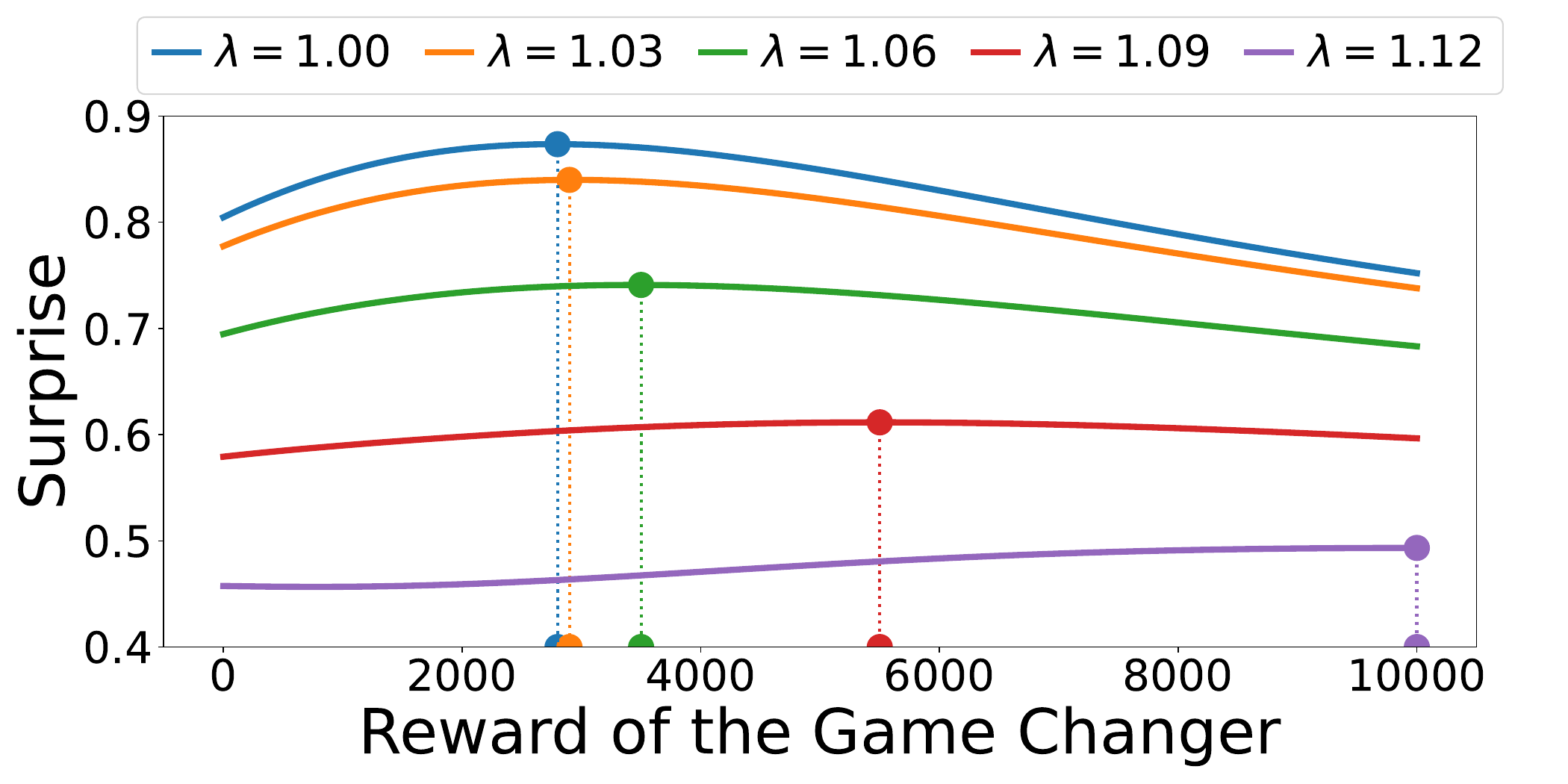}}
    \subfigure[DOTA 2]{\includegraphics[scale=0.127]{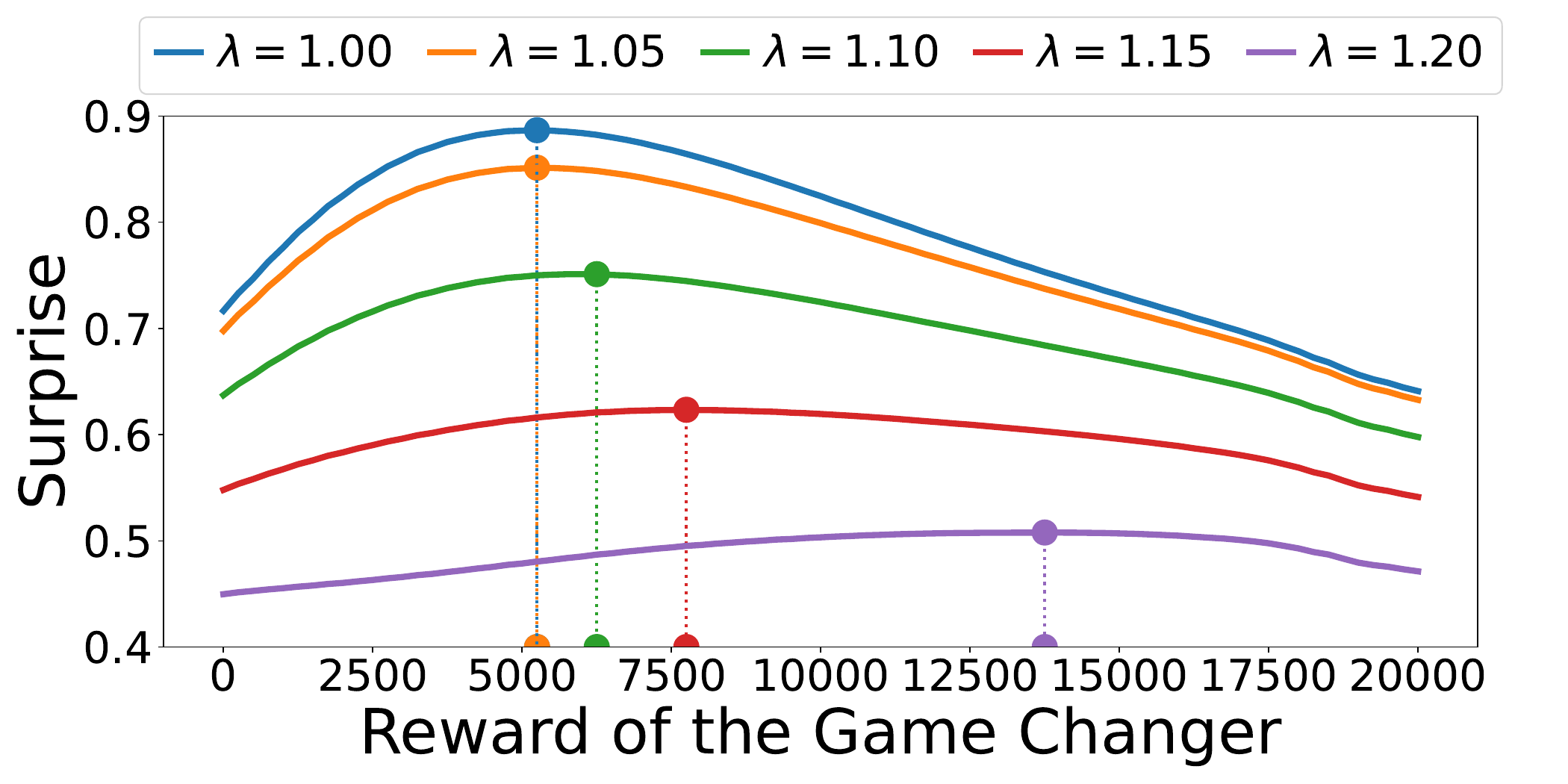}}
    \caption{\textbf{Relationship between the reward of the ``Game Changer'' $\deltaw_{GC}$ and surprise}: When the rating ratio $\lambda$ is larger, i.e., one team has more advantage in gaming skills, the optimal reward $\deltaw_{GC}^*$ is larger.}
    \label{fig:diff_delta}
\end{minipage}
\end{figure}

\subsection{Results of MOBA Games}

We numerically calculate the overall expected surprise given different $\deltaw_{GC}$ in the setting of different $\lambda$. \Cref{fig:diff_delta} shows the numerical results for the optimal reward of ``Game Changer'' in LOL and DOTA 2, which are marked by dots. We have the following observations.
\begin{itemize}
    \item \textbf{Symmetric $\lambda=1$:} the optimal reward is approximately $2800$ for LOL and $5200$ for DOTA 2.
    \item \textbf{General $\lambda\geq 1$:} when the rating ratio $\lambda$ increases, i.e., one team has more advantage in gaming skills, the optimal reward increases.
\end{itemize}

The observation in the general case matches the result for Quidditch. Unlike Quidditch, when two teams have equal ratings, the optimal reward is not $0$ in MOBA. One possible reason is that, the winning probability in MOBA is dynamic. Therefore, a lopsided situation can also happen between two equal rating teams due to advantages at the beginning rounds. We will discuss more in \Cref{sec:con}. 

\section{Conclusion and Discussion}
\label{sec:con}
We study how to design the reward of the ``Game Changer'' to maximize the audiences' expected surprise. We perform a theoretical study for Quidditch and a numerical studies for MOBA games. In the both cases, we find that when the match-up is more unbalanced, the ``Game Changer'' should have a higher reward to allow the underdog to make a comeback.

In MOBA, when the match-up is balanced, the optimal reward is not $0$. One reason is that Quidditch has a fixed winning probability in each round while MOBA has a dynamic winning probability of each teamfights. The team with an early advantage is likely to gain a larger advantage, which is similar to the well-known Matthew effect~\cite{merton1968matthew}. Therefore, a ``Game Changer'' is still necessary in this setting. This is more similar to the uniform prior setting~\cite{huang2021bonus} where one team's winning probability for a single round increases when the team wins more rounds.

One important direction is to extend the result to a more general model, even an unstructured game. In many games, the ``Game Changer's'' reward should be fixed. Here we can learn the diversity of the players first and employ our techniques but weigh over all possible scenarios to provide an optimal design. In some games, it is possible to set the ``Game Changer's'' reward dynamically. Thus, another direction is to study the dynamic setting (e.g. Mario Kart~\cite{website:mario}). Moreover, in our setting, the ``Game Changer'' does not favor any player, while in some settings (e.g. Project Winter~\cite{website:projectwinter}, Mario Kart~\cite{website:mario}), the weaker player can obtain better tools. It is also interesting to model and study those asymmetric settings. 

We do not consider the strategic behaviors of the players in our setting. However, in many games (e.g. Avalon~\cite{website:proavalon}, Mafia), the strategy players adopt also affects the overall surprise of the game. In those settings, a future direction is to combine reinforcement learning tools to learn the optimal strategy and derive a game design based on the strategy.

\clearpage
\bibliographystyle{unsrt}
\bibliography{ref}
\clearpage
\appendix

\section{Numerical results on Surprise Curves}\label{sec:surp_curve}
We also provide numerical results illustrating how the overall surprise depends on the score in \Cref{fig:surp_curve}. We plot the continuous generalization of $\E[\Delta_\belset(x)]$, $\surp(x)$, and annotate the estimation $\tilde{x}$ and the simple estimation $\frac{1}{2q}\left(\frac{1-p}{p}-1\right)$. The overall surprise varies with $x$ and in some cases (\emph{e.g.} $p=0.2,q=0.1$), the optimal $x^*$ creates strictly higher surprise than the trivial settings ($x=0,\infty$). We also observe that when $x$ goes to infinity, the overall surprise becomes smaller since the whole game will reduce to who catches the Golden Snitch. 

\begin{figure}[!ht]\centering
  \subfigure[p=0.5,q=0.1]{\includegraphics[width=.3\linewidth]{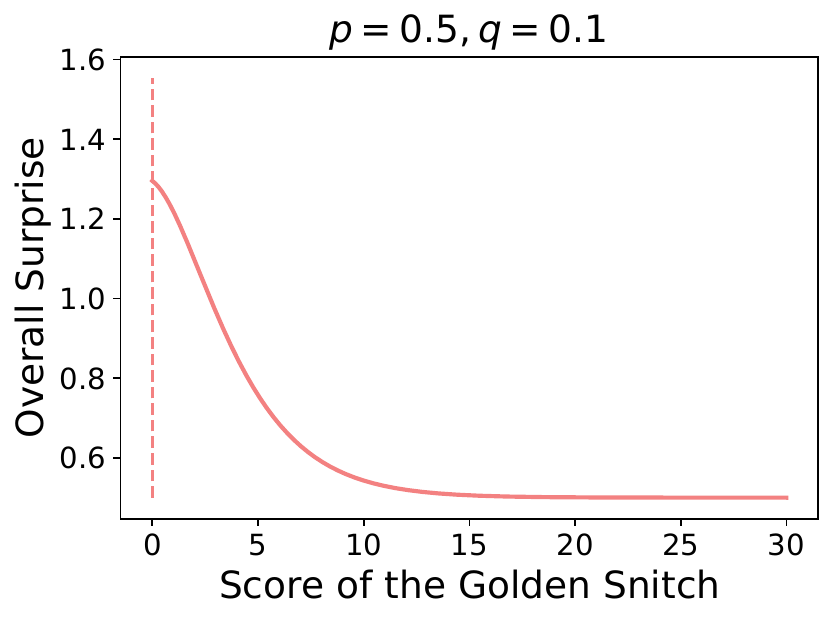}\label{fig:surp_curve_1_1}}
  \subfigure[p=0.3,q=0.1]{\includegraphics[width=.3\linewidth]{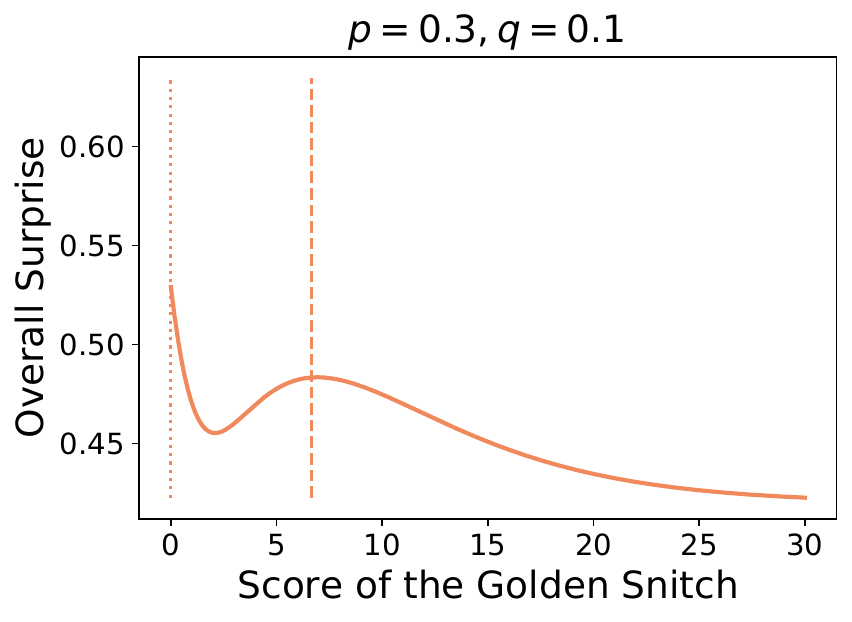}\label{fig:surp_curve_2_1}}
  \subfigure[p=0.2,q=0.1]{\includegraphics[width=.3\linewidth]{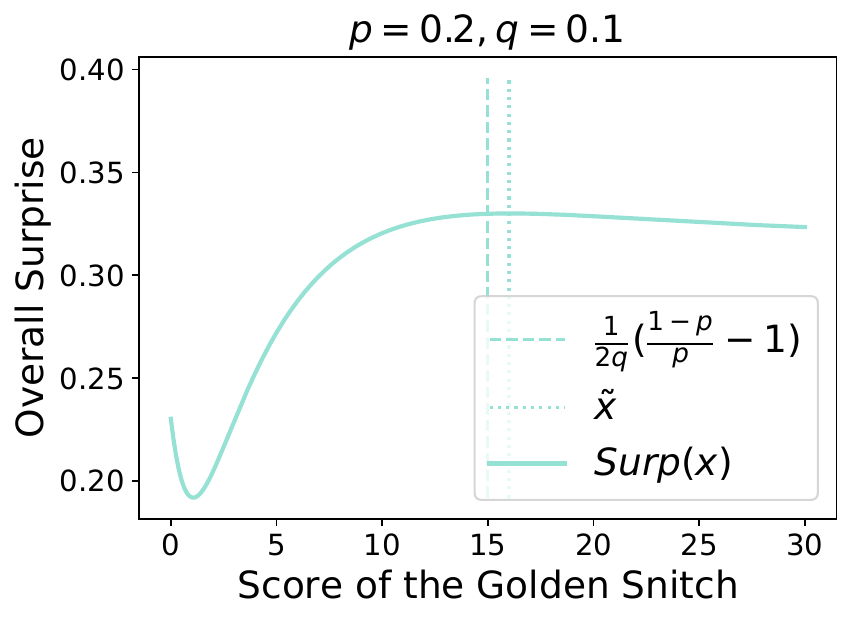}\label{fig:surp_curve_3_1}}
  \subfigure[p=0.5,q=0.2]{\includegraphics[width=.3\linewidth]{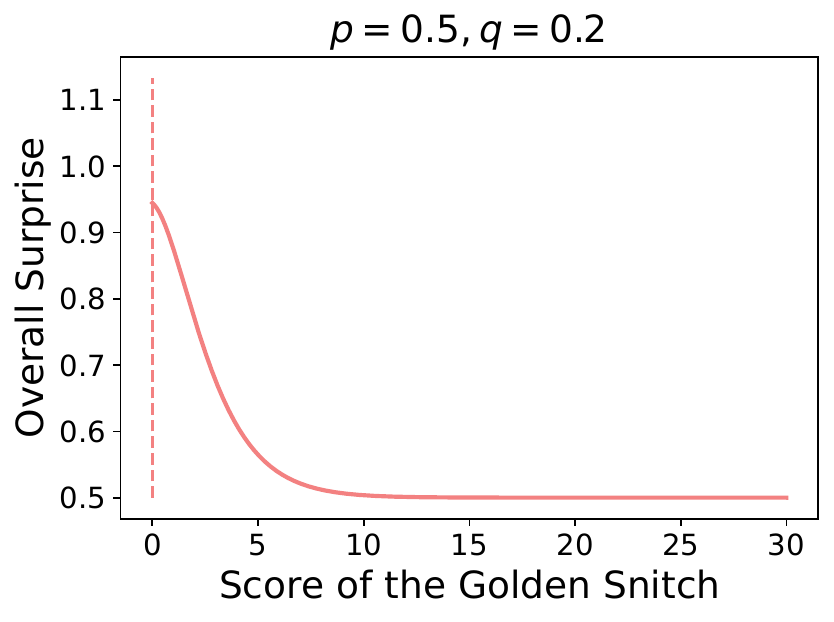}\label{fig:surp_curve_1_2}}
  \subfigure[p=0.3,q=0.2]{\includegraphics[width=.3\linewidth]{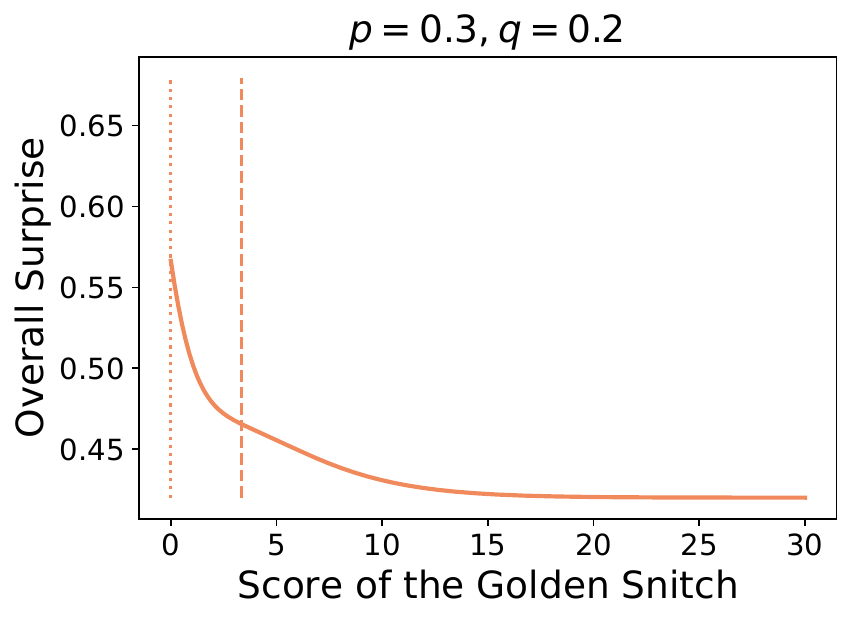}\label{fig:surp_curve_2_2}}
  \subfigure[p=0.2,q=0.2]{\includegraphics[width=.3\linewidth]{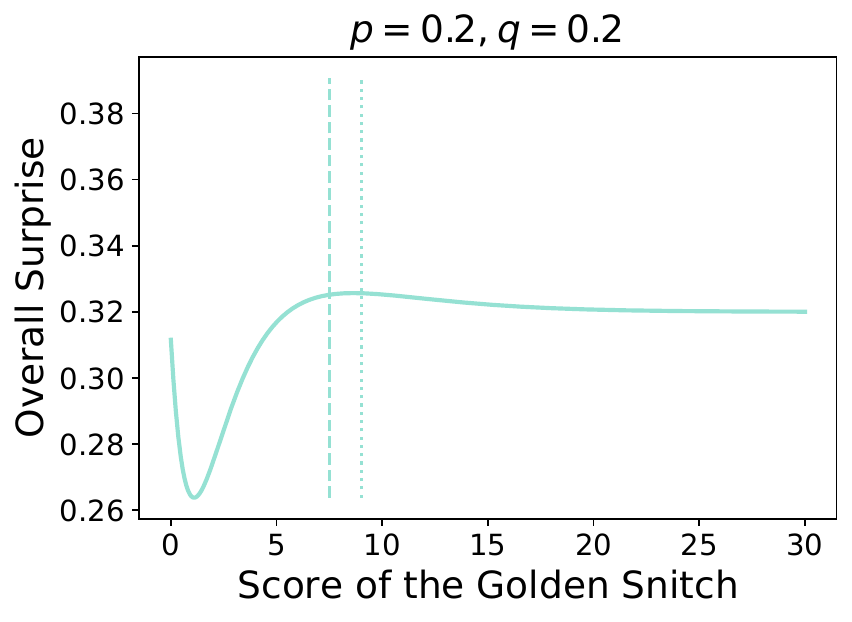}\label{fig:surp_curve_3_2}}
  \subfigure[p=0.5,q=0.3]{\includegraphics[width=.3\linewidth]{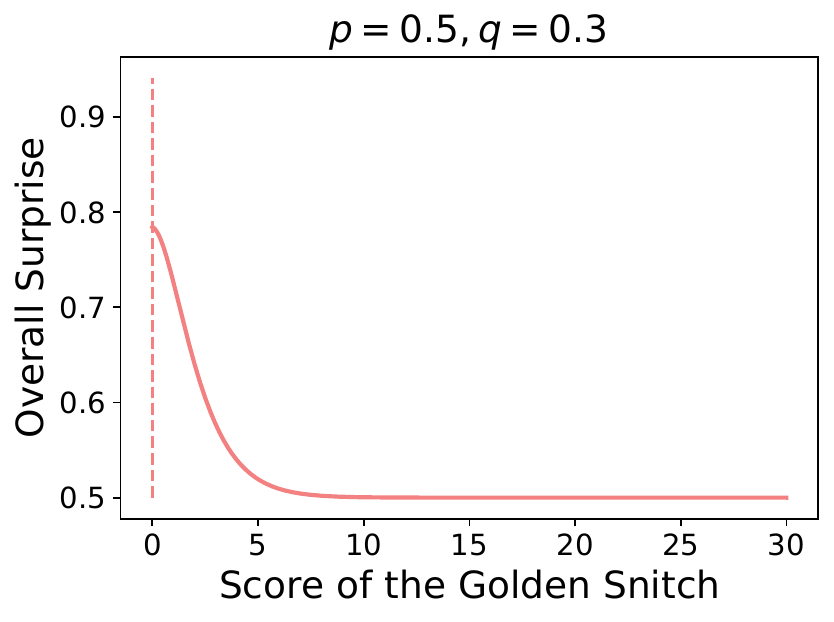}\label{fig:surp_curve_1_3}}
  \subfigure[p=0.3,q=0.3]{\includegraphics[width=.3\linewidth]{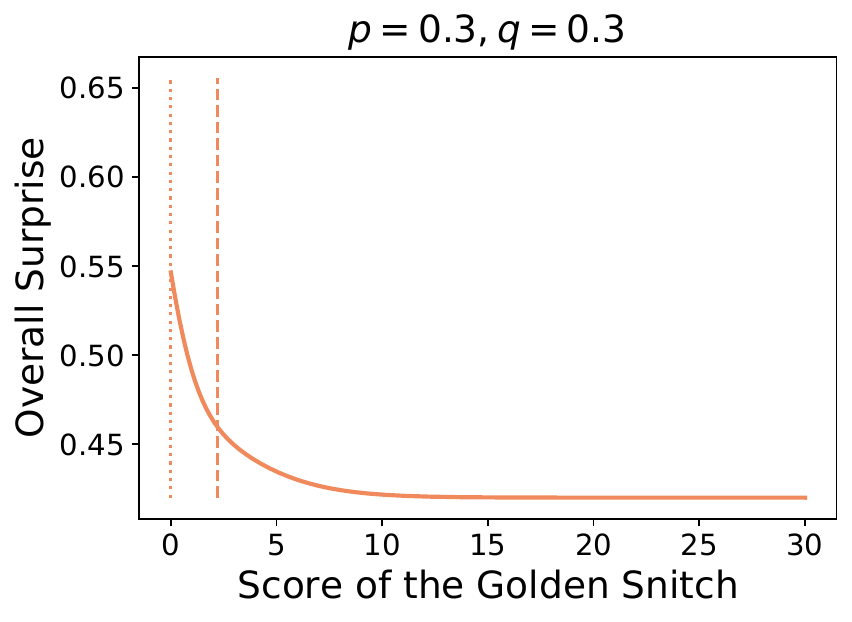}\label{fig:surp_curve_2_3}}
  \subfigure[p=0.2,q=0.3]{\includegraphics[width=.3\linewidth]{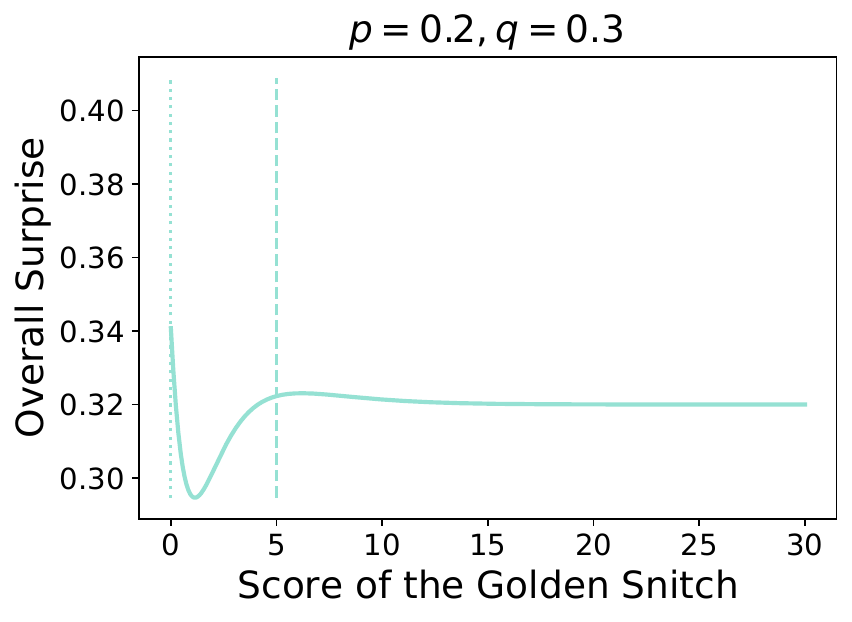}\label{fig:surp_curve_3_3}}
  \caption{\textbf{Relation between Golden Snitch's score and the expected overall surprise.}}
  \label{fig:surp_curve}
\end{figure}

\section{Data Processing}\label{sec:alg}
\subsection{Device}
We conducted our data collection process and numerical experiments on a Intel(R) Core(TM) i5-7267U CPU, with 2 cores and 8 GB RAM. The operating system is Windows 10 21H2. The Python version is 3.7.0. The G++ version is 4.8.1. All data do not contain personally identifiable information or offensive content. 

\subsection{Data Collection}
\paragraph{LOL} We use Games of Legends\footnote{https://gol.gg/esports/home/} as the data source of LOL. For each match, Games of Legends provides a detailed timeline. This timeline includes both teams' wealth changes per minute and all kills that occur in the match per second. The raw data does not include teamfight information. Therefore, we use the following steps to identify teamfights. We partition the kills into clusters. We construct the first cluster as the first $k_1$ kills where 1) for all $k<k+1\leq k_1$, the $k^{th}$ kill and the $k+1^{th}$ kill are within a time gap of less than 30 seconds; 2) the time gap between $k_1^{th}$ kill and the $k_1+1^{th}$ kill are above 30 seconds, and so forth. Kills may happen outside a teamfight (e.g. solo kill). Moreover, when the deaths of both sides are equal in one cluster, two sides will gain equal wealth. This corresponds to the farming round in our model. Therefore, we exclude the clusters where there is only 1 death or the deaths of both sides are equal. The remaining clusters are the set of teamfights in a match. The winner of a teamfight is the team with fewer deaths.

\paragraph{DOTA 2} The DOTA 2 data is collected from OpenDota\footnote{https://www.opendota.com/} via the OpenDota API. Opendota provides match ids for all professional matches. After getting the match ids, we get the data of each match from OpenDota. The data includes the total wealth per player per minute, the time of each teamfight and the number of deaths on both teams, and the time each Roshan was killed. Like LOL, we also exclude the teamfights where the deaths of both sides are equal. The winner of a teamfight is the team with fewer deaths.

In both settings, we label the end time of a teamfight as the time of the last kill in the teamfight. If a round contains the end of at least one teamfight, then this round will be labeled as a teamfight round. Otherwise it will be labeled as a farming round. The winner of a teamfight round is determined by the last teamfight in the round. 

Since we aim to calculate the optimal reward of the ``Game Changer'', we need to remove the influence of the original ``Game Changer'' (\emph{i.e.} Roshan, Baron Nashor), in the data. Therefore, when we calculate the income of each round, we need to subtract the reward brought by the ``Game Changer''. In LOL, Baron Nashor will give each player of the killing team 300 wealth, which is 1500 for the whole team. In DOTA 2, Roshan will reward the killing team 920 wealth in expectation. Besides the wealth reward, Roshan also drops several powerful consumable items on his death.\footnote{https://dota2.fandom.com/wiki/Roshan} We removed the income brought by these 
wealth and items rewards.

\section{Surprise Calculation}
Recall that in our model, the process of MOBA is a Markov chain where each state is described by the current time $t$, the current wealth of the two teams $w_A,w_B$ (at the beginning of the round), and the time that the last ``Game Changer'' is killed. Let $tran(s,t)$ for two states $s,t$ be the probability that state $s$ transfers to state $t$ in the Markov process.

We define $\mathrm{winp}$ of a state $s$ as the winning probability (of the whole game) of team A at state $s$. We define $\mathrm{sur}$ of a state $s$ as the expected surprise that will be generated from the process started at state $s$. We calculate them based on the recurrence relationship between the states.

By definition, we calculate $\mathrm{winp}$ and $\mathrm{sur}$ by
\begin{align*}
\mathrm{winp}(s) & =\sum_{t}\mathrm{tran}(s,t)\mathrm{winp(t)}\\
\mathrm{sur}(s)& =\sum_{t}\mathrm{tran}(s,t)\left(|\mathrm{winp}(s)-\mathrm{winp}(t)|+\mathrm{sur}(t)\right)
\end{align*}
Finally, we calculate the overall expected surprise by backward induction.

To reduce the computational complexity, we reduce the number of states by rounding the wealth to its nearby grid point. To reduce the error introduced by this rounding, when calculating the win rate and surprise for each state, we search forward more than one step. For example, when we calculate the win rate of state $s$ and search forward two steps, we use the following the formula.  

\begin{align*}
\mathrm{winp}(s) & =\sum_{t_1}\sum_{t_2}\mathrm{tran}(s,t_1)\mathrm{tran}(t_1,t_2)\mathrm{winp(t_2)}
\end{align*}

We search forward five steps to reduce the error.

\end{document}